\documentclass{ejm}

\usepackage{hyperref}
\usepackage{caption}

\pdfoutput=1

\newcommand{\HighResFig}[2]{#2} 

\usepackage{mystyle}

\usepackage{MyMnSymbol} 
\usepackage{xcolor}
\usepackage{microtype}
\usepackage{pbox}
\usepackage{epstopdf}

\HighResFig{
\graphicspath{{./images-all/}}
}{
\graphicspath{{./images/}}
}

\newcommand{\shortcite}{\cite}

\usepackage{ejm-additional}

\title[Quantum Optimal Transport for Tensor Field Processing]{Quantum Optimal Transport\\for Tensor Field Processing}

\author[G. Peyr\'e, L. Chizat, F-X. Vialard and J. Solomon]{%
  Gabriel Peyr\'e$\,^1$,\ns
  L\'ena\"{i}c Chizat$\,^2$,\ns
  Fran\c{c}ois-Xavier Vialard$\,^2$ \and
  Justin Solomon$\,^3$
}

\affiliation{%
  $^1\,$CNRS and \'Ecole Normale Sup\'erieure\\
    email\textup{\nocorr: \texttt{gabriel.peyre@ens.fr}}\\
  $^2\,$Univ. Paris-Dauphine, INRIA Mokaplan\\
  $^3\,$MIT}

\begin{document}

\maketitle


\begin{abstract}
This article introduces a new notion of optimal transport (OT) between tensor fields, which are measures whose values are positive semidefinite (PSD) matrices.
This ``quantum'' formulation of OT (Q-OT) corresponds to a relaxed version of the classical Kantorovich transport problem, where the fidelity between the input PSD-valued measures is captured using the geometry of the Von-Neumann quantum entropy.
We propose a quantum-entropic regularization of the resulting convex optimization problem, which can be solved efficiently using an iterative scaling algorithm. This method is a generalization of the celebrated Sinkhorn algorithm to the quantum setting of PSD matrices. 
We extend this formulation and the quantum Sinkhorn algorithm to compute barycenters within a collection of input tensor fields. 
We illustrate the usefulness of the proposed approach on applications to procedural noise generation, anisotropic meshing, diffusion tensor imaging and spectral texture synthesis. 
\end{abstract}

\newcommand{\meshFigIntr}[1]{\includegraphics[width=.100\linewidth,trim=140 10 120 0,clip]{mesh-intro/interpol-#1}}
\newcommand{\imgFigIntr}[1]{\includegraphics[width=.105\linewidth,trim=80 20 80 20,clip]{2d-intro/interpol-ellipses-#1}}

\begin{figure}\centering
\centering
\HighResFig{
\begin{tabular}{@{\hspace{1mm}}c@{\hspace{1mm}}c@{\hspace{1mm}}c@{\hspace{1mm}}c@{\hspace{1mm}}c@{\hspace{1mm}}c@{\hspace{1mm}}c@{\hspace{1mm}}c@{\hspace{1mm}}c@{}}
\imgFigIntr{1}&
\imgFigIntr{2}&
\imgFigIntr{3}&
\imgFigIntr{4}&
\imgFigIntr{5}&
\imgFigIntr{6}&
\imgFigIntr{7}&
\imgFigIntr{8}&
\imgFigIntr{9}\\
\meshFigIntr{1}&
\meshFigIntr{2}&
\meshFigIntr{3}&
\meshFigIntr{4}&
\meshFigIntr{5}&
\meshFigIntr{6}&
\meshFigIntr{7}&
\meshFigIntr{8}&
\meshFigIntr{9}\\
$t=0$ & $t=1/8$ & $t=1/4$ & $t=3/8$ & $t=1/2$ & $t=5/8$ & $t=3/4$ & $t=7/8$ & $t=1$ 
\end{tabular}
}{\includegraphics[width=\linewidth]{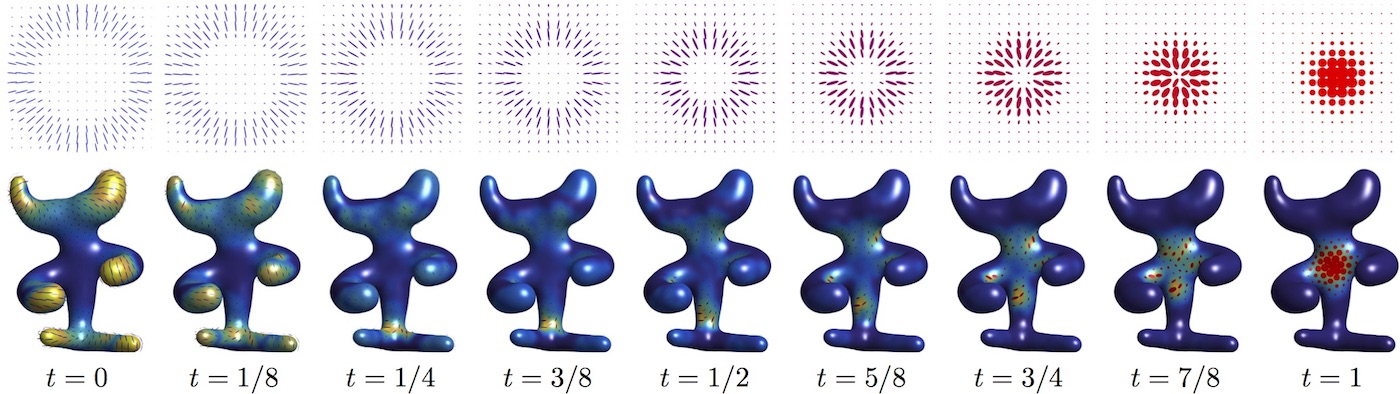}}
\caption{
Given two input fields of positive semidefinite matrices (displayed at times $t \in \{0,1\}$ using ellipses) on some domain (here, a 2-D planar square and a surface mesh), our Quantum Optimal Transport (Q-OT) method defines a continuous interpolating path for $t \in [0,1]$. Unlike linear interpolation schemes, Q-OT transports the ``mass'' of the tensors (size of the ellipses) as well as their anisotropy and orientation. This interpolation, and its extension to finding the barycenter of several input fields, is computed using a fast extension of the well-known Sinkhorn algorithm.
} \label{fig:intro}
\end{figure}


\section{Introduction}
\label{sec-intro}


Optimal transport (OT) is an active field of research at the intersection of probability theory, PDEs, convex optimization and numerical analysis. 
OT offers a canonical way to lift a ground distance on some metric space to a metric between arbitrary probability measures defined over this base space. OT distances offer many interesting  features and in particular lead to a geometrically faithful way to manipulate and interpolate probability distributions.

\subsection{Previous Work}

\subsubsection{Scalar-valued optimal transport.}
Dating back to the eighteenth century, classical instances of the optimal transport problem seek a minimal-cost matching between two distributions defined over a geometric domain, e.g.\ matching supply to demand while incurring minimal cost. 
Initially formulated by Monge in terms of an unknown map transporting mass~\shortcite{Monge1781}, its reformulation by Kantorovich~\shortcite{Kantorovich42} as a linear program (static formulation) enables the use of convex analysis to study its structure and develop numerical solvers. 
The equivalence between these two formulations was introduced by Brenier~\shortcite{Brenier91} and opened the door to a dynamical (geodesic) reformulation~\cite{benamou2000computational}. We refer to~\cite{santambrogio2015optimal} for a review of the theoretical foundations of OT. 

The basic OT problem has been extended in various ways, a typical illustration of which being the computation of a barycenter (Fr\'echet mean) of input measures, a convex program studied by Agueh and Carlier~\shortcite{agueh-2011}.
OT also has found numerous applications, for instance in computer vision (under the name ``earth mover's distance'')~\cite{rubner-2000} and computer graphics~\cite{bonneel-2011}. 

\subsubsection{Unbalanced transport.}

While the initial formulations of OT are restricted to positive measures of equal mass (normalized probability distributions), a recent wave of activity has proposed and studied a family of ``canonical'' extensions 
to the \emph{unbalanced} setting of arbitrary positive measures. This covers both a dynamic formulation~\cite{LieroMielkeSavareCourt,kondratyev2015,2016-chizat-focm} and a static one~\cite{LieroMielkeSavareLong,2015-chizat-unbalanced} and has been applied in machine learning~\cite{frogner-2015}. 
Our work extends this static unbalanced formulation to tensor-valued measures.

\subsubsection{Entropic regularization.}

The current state-of-the-art OT approximation for arbitrary ground costs uses entropic regularization of the transport plan. This leads to strictly convex programs that can be solved using a simple class of highly parallelizable diagonal scaling algorithms. The landmark paper of Cuturi~\shortcite{cuturi-2013} inspired detailed study of these solvers, leading to various generalizations of Sinkhorn's algorithm~\shortcite{Sinkhorn64}. This includes for instance the use of fast convolutional structures~\cite{solomon-2015}, extensions to barycenters~\cite{benamou-2015} and application to unbalanced OT~\cite{frogner-2015,2016-chizat-sinkhorn}.
These entropic regularization techniques correspond to the use of projection and proximal maps for the Kullback--Leibler Bregman divergence and are equivalent to iterative projections~\cite{bregman-1967} and Dykstra's algorithm~\cite{Dykstra83,bauschke-lewis}. 
An important contribution of the present work is to extend these techniques to the matrix setting (i.e., using quantum divergences). Note that quantum divergences have been recently used to solve some machine learning problems~\cite{Dhillon2008,Kulis2009,Chandrasekaran2017}.

\subsubsection{Tensor field processing.}

Tensor-valued data are ubiquitous in various areas of imaging science, computer graphics and vision. In medical imaging, diffusion tensor imaging (DTI)~\cite{wandell2016clarifying} directly maps observed data to fields of tensors, and specific processing methods have been developed for this class of data (see e.g.~\cite{Dryden2009,Deriche2006}). Tensor fields are also at the heart of anisotropic diffusion techniques in image processing~\cite{weickert1998anisotropic}, anisotropic meshing~\cite{alliez2003anisotropic,demaret2006image,peyre-iccv-09}, and anisotropic texture generation~\cite{LagaImproving}; they also find application in line drawing~\cite{VaxmanCDPBHB16} and data visualization~\cite{HotzFHHJJ04}. 

\subsubsection{OT and Sinkhorn on pairs of tensors.}

Although this is not the topic of this paper, we note that several 
notions of OT have been defined between two tensors \emph{without} any spatial displacement. Gurvits introduced in~\cite{gurvits2004classical} a Sinkhorn-like algorithm to couple two tensors by an endomorphism that preserves positivity. This algorithm, however, is only known to converge in the case where the two involved tensors are the identity matrices;  see~\cite{georgiou2015positive} for a detailed analysis. 
In contrast to this ``static'' formulation that seeks for a joint coupling between the tensors, a geodesic dynamic formulation is proposed in~\cite{Carlen2014}; see also~\cite{Chen2016,ChenGangbo17} for a related approach.

\subsubsection{OT on tensor fields.}

The simplest way to define OT-like distances between arbitrary vector-valued measures is to use dual norms~\cite{Ning2014metrics}, which correspond to generalizations of $W_1$ OT for which transport cost equals ground distance. The corresponding metrics, however, have degenerate behavior in interpolation and barycenter problems (much like the $L^1$ norm on functions) and only use the linear structure of matrices rather than their multiplicative structure.
More satisfying notions of OT have recently been proposed in a dynamical (geodesic) way~\cite{Chen2016,JiangSpectral}. 
A static formulation of a tensor-valued OT is proposed in~\cite{ning2015matrix}, but it differs significantly from ours. It is initially motivated using a lifting that squares the number of variables, but a particular choice of cost reduces the computation to the optimization of a pair of couplings. In contrast, the formulation we propose in the present article is a direct generalization of unbalanced OT to matrices, which in turn enables the use of a Sinkhorn algorithm.




\subsection{Contributions} 

We present a new static formulation of OT between tensor fields, which is the direct generalization of unbalanced OT from the scalar to the matrix case.
Our second contribution is a fast entropic scaling algorithm generalizing the celebrated Sinkhorn iterative scheme. This leads to a method to compute geometrically-faithful interpolations between two tensor fields. 
Our third contribution is the extension of this approach to compute barycenters between several tensor fields. 
The Matlab code to reproduce the results of this article is available online.\footnote{\url{https://github.com/gpeyre/2016-wasserstein-tensor-valued}}

%

\subsection{Notation}

In the following, we denote $\Ss^d \subset \RR^{d \times d}$ the space of symmetric matrices, $\Ss_+^d$ the closed convex cone of positive semidefinite matrices, and $\Ss_{++}^d$ the open cone of positive definite matrices. 
We denote $\exp: \Ss^d \to \Ss_{++}^d$ the matrix exponential, which is defined as $\exp(P)=U \diag_s(e^{\si_s})U^\top$ where $P=U \diag_s(\si_s) U^\top$ is an eigendecomposition of $P$.
We denote $\log: \Ss_{++}^d \to \Ss^d$ the matrix logarithm $\log(P)=U \diag_s(\log \si_s)U^\top$, which is the inverse of $\exp$ on $\Ss_{++}^d$. 
We adopt some conventions in order to deal conveniently with singular matrices.
We extend $(P,Q) \mapsto P \log (Q)$ by lower semicontinuity on $(\Ss_+^d)^2$, i.e.\ writing $Q = U \diag_s(\si_s) U^\top$ and $\tilde{P}=U^\top P U$, 
\eq{
P \log Q :=
\begin{cases}
P \log Q &\text{if $\ker Q=\emptyset$},\\
U[\tilde{P}\diag_s(\log \si_s)]U^\top & \text{if $\ker Q \subset \ker P $,}\\
+\infty & \text{otherwise,}
\end{cases}
}
with the convention $0\log 0 = 0$ when computing the matrix product in square brackets. Moreover, for $(P,(Q_i)_{i\in I}) \in\Ss^d \times (\Ss_+^d)^I$, the matrix $\exp(P + \sum_i \log Q_i)$ is by convention the matrix in $\Ss_+^d$ which kernel is $\sum_i \ker Q_i$ (and is unambiguously defined on the orthogonal of this space).

\begin{figure}\centering
\begin{minipage}[c]{.48\linewidth}
\begin{tabular}{@{}c@{}c@{}}
{\includegraphics[width=1\linewidth]{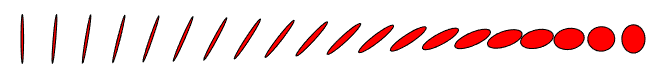}}\\
$X=[0,1]$, $d=2$ \\[2mm]
{\includegraphics[width=1\linewidth]{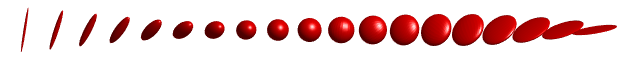}} \\
$X=[0,1]$, $d=3$
\end{tabular}
\end{minipage} 
\begin{minipage}[c]{.49\linewidth}
\begin{tabular}{@{}c@{}c@{}}
\includegraphics[width=.55\linewidth]{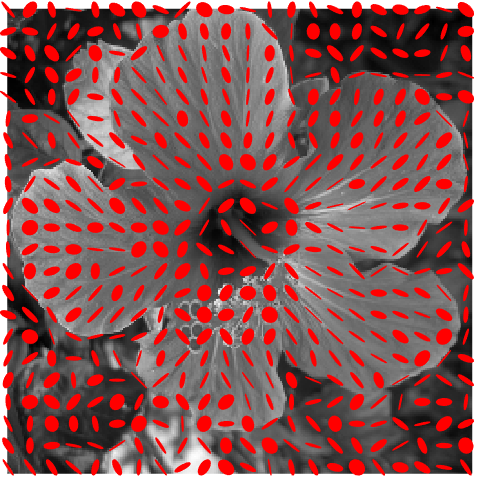}&
\includegraphics[width=.42\linewidth]{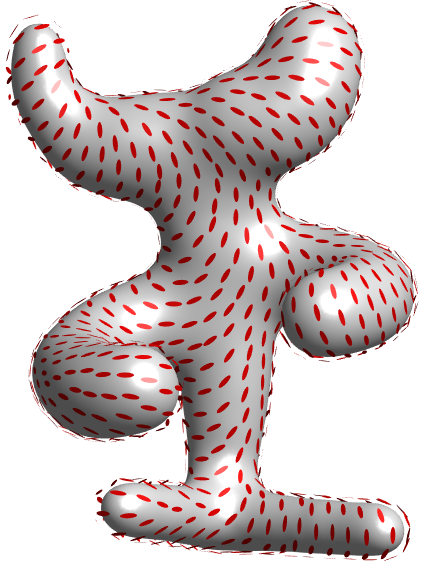}\\
 $X=[0,1]^2$ & $X=$surface \\
\end{tabular}
   \end{minipage}
\caption{Displays of various types of tensor-valued measures $\mu$. The principal directions of an ellipse at some $x_i \in X$ are the eigenvectors of $\mu_i \in \Ss_+^d$, while the principal widths are given by its eigenvalues. 
} \label{fig:display}
\end{figure}

A tensor-valued measure $\mu$ defined on some space $X$ is a vector-valued measure, where the ``mass'' $\mu(A) \in \Ss_+^d$ associated to a measurable set $A \subset X$ is a PSD matrix. In this article, in order to derive computational schemes, we focus on discrete measures. Such a measure $\mu$ is a sum of Dirac masses
$\mu = \sum_{i \in I} \mu_i \de_{x_i}$
where $(x_i)_i \subset X$, and $(\mu_i)_i \in \Ss_+^d$ is a collection of PSD matrices. In this case, $\mu(A)=\sum_{x_i \in A} \mu_i$. 
Figure~\ref{fig:display} shows graphically some examples of tensor-valued measures; we use this type of visualization through the article. 
In the following, since the sampling points $(x_i)_i$ are assumed to be fixed and clear from the context, to ease readability, we do not make the distinction between the measure $\mu$ and the collection of matrices $(\mu_i)_i$. This is an abuse of notation, but it is always clear from context whether we are referring to a measure or a collection of matrices. 

The quantum entropy (also called von Neumann entropy) of a tensor-valued measure is 
\eql{\label{eq-h-quantum}
	H(\mu) \eqdef \sum_i H(\mu_i)
	\qwhereq 
}
\eq{
	\foralls P \in \Ss^d, \quad
	H(P) \eqdef -\tr( P \log(P)-P) - \iota_{\Ss_{+}^d}(P),
}
where $\iota_C$ is the indicator function of a closed convex set $C$, i.e. $\iota_C(P)=0$ if $P \in C$ and $\iota_C(P)=+\infty$ otherwise.
Note that $H$ is a concave function. 
The quantum Kullback-Leibler divergence (also called quantum relative entropy) is the Bregman divergence associated to $-H$. For a collection of PSD matrices $\mu=(\mu_i)_i, \xi=(\xi_i)_i$ in $\Ss_+^d$ corresponding to measures defined on the same grid, it is defined as 
\eql{\label{eq-kl-quantum}
	\KL(\mu|\xi) \eqdef \sum_i \KL(\mu_i|\xi_i), 
}
where for all $(P,Q) \in \Ss_{+}^d \times \Ss_{+}^d$, we denote
\eq{	
	\KL(P|Q) \eqdef 
	\tr( P \log (P) - P\log (Q) - P + Q ) + \iota_{\Ss_{++}^d}(P)
}
which is convex with respect to both arguments. 
The inner product between collections of matrices $\mu=(\mu_i)_i, \xi=(\xi_i)_i$ is 
\eq{
	\dotp{\mu}{\xi} \eqdef \sum_{i} \dotp{\mu_i}{\xi_i} \eqdef \sum_{i} \tr( \mu_i \xi_i^\top ).
}
Given a collection of matrices $\ga=(\ga_{i,j})_{i \in I,j \in J}$ the marginalization operators read
\eq{
	\ga \ones_J \eqdef \Big(\sum_j \ga_{i,j}\Big)_i 
	\qandq
	\ga^\top \ones_I \eqdef \Big(\sum_i \ga_{i,j}\Big)_j .
}


\section{Kantorovich Problem for Tensor-Valued Transport}

We consider two measures that are sums of Dirac masses
\eql{\label{eq-input-measures}
	\mu = \sum_{i \in I} \mu_i \de_{x_i}
	\qandq
	\nu = \sum_{j \in J} \nu_j \de_{y_j}
}
where $(x_i)_i \subset X$ and $(y_j)_j \subset Y$, and $(\mu_i)_i \in \Ss_+^d$ and $(\nu_j)_j \in \Ss_+^d$ are collections of PSD matrices. 
Our goal is to propose a new definition of OT between $\mu$ and $\nu$.

\subsection{Tensor Transportation}
\label{sec-tensor-ot}

Following the initial static formulation of OT by Kantorovich~\shortcite{Kantorovich42}, we define a coupling $\ga = \sum_{i,j}\ga_{i,j} \de_{(x_i,y_j)}$ as a measure over the product $X \times Y$ that encodes the transport of mass between $\mu$ and $\nu$. In the matrix case, $\ga_{i,j} \in \Ss_+^d$ is now a PSD matrix, describing how much mass is moved between $\mu_i$ and $\nu_j$. 
Exact (balanced) transport would mean that the marginals $(\ga \ones_J,\ga^\top \ones_I)$ must be equal to the input measures $(\mu,\nu)$. But as remarked by Ning et al.~\shortcite{ning2015matrix}, in contrast to the scalar case, in the matrix case (dimension $d>1$), this constraint is in general too strong, and there might exist no coupling satisfying these marginal constraints.
We advocate in this work that the natural workaround for the matrix setting is the unbalanced case, and following~\cite{LieroMielkeSavareLong}, we propose to use a ``relaxed'' formulation where the discrepancy between the marginals $(\ga \ones_J,\ga^\top \ones_I)$ and the input measures $(\mu,\nu)$ is quantified according to some divergence between measures. 

In the scalar case, the most natural divergence is the Kulback-Leibler divergence (which in particular gives rise to a natural Riemannian structure on positive measures, as defined in~\cite{LieroMielkeSavareCourt,kondratyev2015,2016-chizat-focm}).  We propose to make use of its quantum counterpart~\eqref{eq-kl-quantum} 
%
via
  the following convex program
\eql{\label{eq-Kantorovich}
	W(\mu,\nu) = \min_{\ga}\ \dotp{\ga}{c} + \rho_1 \KL(\ga \ones_J|\mu) + \rho_2 \KL(\ga^\top \ones_I|\nu) 
}
subject to the constraint $\foralls (i,j), \ga_{i,j} \in \Ss_+^d$.
Here $\rho_1,\rho_2 >0$ are constants balancing the ``transport'' effect versus the local modification of the matrices. 

The matrix $c_{i,j} \in \RR^{d \times d}$ measures the cost of displacing an amount of (matrix) mass $\ga_{i,j}$ between $x_i$ and $y_j$  as $\tr(\ga_{i,j}c_{i,j})$. 
A typical cost, assuming $X=Y$ is a metric space endowed with a distance $d_X$, is
\eq{
	c_{i,j} = d_X(x_i,y_j)^\al \Id_{d \times d}, 
}
for some $\al>0$.  In this case, one should interpret the trace as the global mass of a tensor, and the total transportation cost is simply 
\eq{
	\dotp{\ga}{c} = \sum_{i,j} d_X(x_i,y_j)^\al \tr(\ga_{i,j}).
}

\begin{rem}[Classical OT]\label{rem-classical-ot}
	In the scalar case $d=1$, \eqref{eq-Kantorovich} recovers exactly the log-entropic definition~\cite{LieroMielkeSavareLong} of unbalanced optimal transport, which is studied numerically by Chizat et al.~\shortcite{2016-chizat-sinkhorn}. 
	For isotropic tensors, i.e., all $\mu_i$ and $\nu_j$ are scalar multiples of the identity $\Id_{d \times d}$, the computation also collapses to the scalar case (the $\ga_{i,j}$ are also isotropic). More generally, if all the $(\mu_i,\nu_j)_{i,j}$ commute, they diagonalize in the same orthogonal basis, and~\eqref{eq-Kantorovich} reduces to performing $d$ independent unbalanced OT computations along each eigendirection. 
\end{rem}

\begin{rem}[Cost between single Dirac masses]
	When $\mu=P \de_{x}$ and $\nu=Q \de_{x}$ are two Dirac masses at the same location $x$ and associated tensors $(P,Q) \in (\Ss_{+}^d)^2$,
	one obtains the following ``metric'' between tensors (assuming $\rho_1=\rho_2=1$ for simplicity)
	\eql{\label{eq-cost-single-dirac}
		\sqrt{W(P \de_{x},Q \de_{x})} = D(P,Q) \eqdef \tr\pa{
			P+Q-2 \mathfrak{M}(P,Q)
		}^{\frac{1}{2}}
	}	
	where $\mathfrak{M}(P,Q) \eqdef \exp(\log(P)/2+\log(Q)/2)$.
	When $(P,Q)$ commute, one has $D(P,Q) = \norm{\sqrt{P}-\sqrt{Q}}$ which is a distance.
	In the general case, we do not know whether $D$ is a distance (basic numerical tests do not exclude this property).
	%
	
%
	%
	%
\end{rem}

\begin{rem}[Quantum transport on curved geometries]
	If $(\mu,\nu)$ are defined on a non-Euclidean space $Y=X$, like a smooth manifold, then formulation~\eqref{eq-Kantorovich} should be handled with care, since it assumes all the tensors $(\mu_i,\nu_j)_{i,j}$ are defined 
	in some common basis. 
	For smooth manifolds, the simplest workaround is to assume that these tensors are defined with respect to carefully selected orthogonal bases of the tangent planes, so that the field of bases is itself smooth. Unless the manifold is parallelizable, in particular if it has a trivial topology, it is not possible to obtain a globally smooth orthonormal basis; in general, any such field necessarily has a few singular points. In the following, we compute smoothly-varying orthogonal bases of the tangent planes (away from singular points) following the method of Crane et al.~\shortcite{crane2010trivial}. 
	In this setting, the cost is usually chosen to be $c_{i,j} = d_X(x_i,x_j)^\al \Id_{d \times d}$ where $d_X$ is the geodesic distance on $X$. 
\end{rem}

\begin{rem}[Measure lifting]
An alternative to compute OT between tensor fields would be to rather lift the input measure $\mu$ to a measures $\bar \mu \eqdef \sum_{i \in I}  \de_{(\mu_i,x_i)}$ defined over the space $X \times \Ss_+^d$ (and similarly for the lifting $\bar\nu$ of $\nu$) and then use traditional scalar OT over this lifted space (using a ground cost taking into account both space and tensor variations). 
Such a naive approach would destroy the geometry of tensor-valued measures (which corresponds to the topology of weak convergence of measures), and result in very different interpolations. For example, a sum of two nearby Diracs on $X=\RR$  
\eq{
  \mu = P \delta_0 + Q \delta_s   \qwhereq   
  P \eqdef \begin{pmatrix}1 & 0 \\ 0 & 0\end{pmatrix} 
  \qandq 
  Q \eqdef \begin{pmatrix}0 & 0\\0 & 1\end{pmatrix}
}
is treated by our method as being very close to $\Id_{2\times 2} \de_{0}$ (which is the correct behaviour of a \emph{measure}), whereas it would be lifted to 
$\bar\mu=\de_{(0,P)}+\de_{(s,Q)}$ over $\RR \times \Ss_2^+$, which is in contrast very far from $\de_{(0,\Id_{2\times 2})}$.
\end{rem}

\subsection{Quantum Transport Interpolation}

\newcommand{\muA}{\mu}
\newcommand{\muB}{\nu}

Given two input measures $(\muA,\muB)$, we denote by $\ga$ a solution of~\eqref{eq-Kantorovich} or, in practice, its regularized version (see~\eqref{eq-Kantorovich-regul} below). The coupling $\ga$ defines a (fuzzy) correspondence between the tensor fields. A typical use of this correspondence is to compute a continuous interpolation between these fields. Section~\ref{sec-numerics-interp} shows some numerical illustrations of this interpolation. Note also that Section~\ref{sec-q-bary} proposes a generalization of this idea to compute an interpolation (barycenter) between more than two input fields.

Mimicking the definition of the optimal transport interpolation (the so-called McCann displacement interpolation; see for instance~\cite{santambrogio2015optimal}), we propose to use $\gamma$ to define a path $t \in [0,1] \mapsto \mu_t$ interpolating between $(\muA,\muB)$. 
For simplicity, we assume the cost has the form $c_{i,j}=d_X(x_i,y_j)^\al \Id_{d \times d}$ for some ground metric $d_X$ on $X=Y$. We also suppose we can compute efficiently the interpolation between two points $(x_i,y_j) \in X^2$ as
\eq{
	x_{i,j}^t \eqdef \uargmin{x \in X} (1-t)d_X^2(x_i,x) + t d_X^2(y_j,x). 
}
For instance, over Euclidean spaces, $g_t$ is simply a linear interpolation, and over more general manifold, it is a geodesic segment.
We also denote
\eq{
	\bar\muA_i \eqdef \muA_i \Big( \sum_{j} \ga_{i,j} \Big)^{-1} 
	\qandq
	\bar\muB_j \eqdef \muB_j \Big( \sum_{i} \ga_{i,j} \Big)^{-1}
}
the adjustment factors which account for the imperfect match of the marginal associated to a solution of~\eqref{eq-Kantorovich-regul}; the adjusted coupling is
\eq{
	\ga_{i,j}^t \eqdef [(1-t) \bar\muA_i + t \bar\muB_{j}] \ga_{i,j}.
}

Finally, the interpolating measure is then defined as
\eql{\label{eq-interpolating}
	\foralls t \in [0,1], \quad
	\mu_t \eqdef \sum_{i,j} \ga_{i,j}^t \de_{x_{i,j}^t}.
}
One easily verifies that this measure indeed interpolates the two input measures, i.e. 
$(\mu_{t=0},\mu_{t=1})=(\muA,\muB)$. 
This formula~\eqref{eq-interpolating} generates the interpolation by creating a Dirac tensor $ \ga_{i,j}^t \de_{x_{i,j}^t}$ for each coupling entry $\ga_{i,j}$, and this tensor travels between $\mu_i \de_{x_i}$ (at $t=0$) and $\nu_j \de_{y_j}$ (at $t=1$).

\begin{rem}[Computational cost] We observed numerically that, similarly to the scalar case, the optimal coupling $\ga$ is sparse, meaning that only of the order of $O(|I|)$ non-zero terms are involved in the interpolating measure~\eqref{eq-interpolating}. Note that the entropic regularization algorithm detailed in Section~\ref{eq-q-sink} destroys this exact sparsity, but we found numerically that thresholding to zero the small entries of $\ga$ generates accurate approximations. 
\end{rem}

\section{Quantum Sinkhorn}
\label{eq-q-sink}

The convex program~\eqref{eq-Kantorovich} defining quantum OT is computationally challenging because it can be very large scale (problem size is $|I| \times |J|$) for imaging applications, and it involves matrix exponential and logarithm. In this section, leveraging recent advances in computational OT initiated by Cuturi~\shortcite{cuturi-2013}, we propose to use a similar entropy regularized strategy (see also section~\ref{sec-intro}), but this time with the quantum entropy~\eqref{eq-h-quantum}. 

\subsection{Entropic Regularization}

We define an entropic regularized version of~\eqref{eq-Kantorovich}
\eql{\label{eq-Kantorovich-regul}
	W_\epsilon(\mu,\nu) \eqdef \min_{\ga} \dotp{\ga}{c} + \rho_1\KL(\ga \ones_J|\mu) + \rho_2\KL(\ga^\top \ones_I|\nu) - \epsilon H(\ga). 
}
Note that when $\epsilon=0$, one recovers the original problem~\eqref{eq-Kantorovich}. 
This is a strongly convex program, with a unique solution. The crux of this approach, as already known in the scalar case (see~\cite{2016-chizat-sinkhorn}), is that its convex dual has a particularly simple structure, which is amenable to a simple alternating maximization strategy. 

\begin{prop}
	The dual problem associated to~\eqref{eq-Kantorovich-regul} reads
	\begin{multline}\label{eq-dual-pb}
		W_\epsilon(\mu,\nu)
		= 
		\umax{u,v} -\!
				\tr\Big[
						\rho_1 \sum_i (e^{u_i+\log(\mu_i)} - \mu_i) \\
					+   \rho_2 \sum_j (e^{v_j+\log(\nu_j)} - \nu_j)
					+    \epsilon \sum_{i,j}  e^{ \Kern(u,v)_{i,j} }
			 \Big], 
	\end{multline}	
	where $u=(u_i)_{i \in I}, v=(v_j)_{j \in J}$ are collections of arbitrary symmetric (not necessarily in $\Ss_+^d$) matrices $u_i, v_j \in \Ss^d$, 
	where we define
	\eql{\label{eq-defn-K}
		\Kern(u,v)_{i,j} \eqdef -\frac{c_{i,j} + \rho_1 u_i + \rho_2 v_j}{\epsilon}.
	} 
	Furthermore, the following primal-dual relationships hold at optimality:
	\eql{\label{eq-primal-dual-scaling}
		\foralls (i,j), \quad \ga_{i,j} = \exp\pa{ \Kern(u,v)_{i,j} }.
	}
\end{prop}
\begin{proof} 
Applying the Fenchel--Rockafellar duality theorem~\cite{rockafellar-convex} to~\eqref{eq-Kantorovich-regul} leads to the dual program
\eq{
	\umax{u,v} - \epsilon \KL^*( \Kern_0(u,v)|\xi) 
	-\rho_1 \KL^*(u|\mu) - \rho_2 \KL^*(v|\nu) - \epsilon \tr(\xi), 
}
where here $\KL^*(\cdot|\mu)$ corresponds to the Legendre transform with respect to the first argument of the KL divergence, $ \Kern_0(u,v)_{i,j} \eqdef -\frac{\rho_1 u_i + \rho_2 v_j}{\epsilon}$ and $\xi_{i,j} \eqdef \exp(-c_{i,j}/\epsilon)$ for all $i,j$.
The following Lengendre formula leads to the desired result:
\eq{
	\KL^*(u|\mu) = \textstyle\sum_i \tr( \exp(u_i+\log(\mu_i)) - \mu_i).
}
\end{proof}

\subsection{Quantum Sinkhorn Algorithm}

It is possible to use Dykstra's algorithm~\shortcite{Dykstra83} (see~\cite{bauschke-lewis} for its extension to Bregman divergences) to solve~\eqref{eq-dual-pb}. This corresponds to alternatively maximizing~\eqref{eq-dual-pb} with respect to $u$ and $v$. 
The following proposition states that the maximization with respect to either $u$ or $v$ leads to two fixed-point equations. 
These fixed points are conveniently written using the log-sum-exp operator, 
\eql{\label{eq-dfn-lse}
	\LSE_j( K ) \eqdef \Big( \log \sum_j \exp(K_{i,j}) \Big)_i, 
}
where the sum on $j$ is replaced by a sum on $i$ for $\LSE_i$. 

\begin{prop}\label{prop-fixed-points}
	For $v$ fixed (resp.\ $u$ fixed), the minimizer $u$ (resp.\ $v$) of~\eqref{eq-dual-pb} satisfies
	\begin{align}\label{eq-fixed-point-u}
		\foralls i, \quad u_i = \LSE_j(\Kern(u,v))_i-\log(\mu_i), \\
		\foralls j, \quad v_j = \LSE_i(\Kern(u,v))_j-\log(\nu_j), \label{eq-fixed-point-v}
	\end{align}
	where $\Kern(u,v)$ is defined in~\eqref{eq-defn-K}.
\end{prop}
\begin{proof}
	Writing the first order condition of~\eqref{eq-dual-pb} with respect to each $u_i$ leads to
	\eq{
		\rho_1 e^{u_i + \log(\mu_i)} - \rho_1 \sum_{j} e^{\Kern(u,v)_{i,j}} = 0
	} 
	which gives the desired expression. A similar expression holds for the first order conditions with respect to $v_j$.
\end{proof}

A simple fixed point algorithm is then obtained by replacing the explicit alternating minimization with respect to $u$ and $v$ in Dykstra's with just one step of fixed point iteration~\eqref{eq-fixed-point-u} and~\eqref{eq-fixed-point-v}. To make the resulting fixed point contractant and ensure linear convergence, one introduces relaxation parameters $(\tau_1,\tau_2)$. 

The quantum Sinkhorn algorithm is detailed in Algorithm~\ref{alg:sinkhorn}. It alternates between the updates of $u$ and $v$, using relaxed fixed point iterations associated to~\eqref{eq-fixed-point-u} and~\eqref{eq-fixed-point-v}. We use the following $\tau$-relaxed assignment notation 
\eql{\label{eq-dfn-relaxed-assign}
	a \RelaxAssign{\tau} b 
	\quad\text{means that}\quad
	a \leftarrow (1-\tau) a + \tau b.
}
The algorithm outputs the scaled kernel $\ga_{i,j} = \exp(K_{i,j})$.

\begin{rem}[Choice of $\tau_k$]\label{rem-choice-tau}
 In the scalar case, i.e. $d=1$ (and also for isotropic input tensors), when using $\tau_k = \tfrac{\epsilon}{\rho_k+\epsilon}$ for $k=1,2$, one retrieves exactly Sinkhorn iterations for unbalanced transport as described in~\cite{2016-chizat-sinkhorn}, and each update of $u$ (resp.\ $v$) exactly solves the fixed point~\eqref{eq-fixed-point-u} (resp.\ \eqref{eq-fixed-point-v}). 
Moreover, it is simple to check that these iterates are contractant whenever
\eq{
	\tau_k \in ]0,\tfrac{2 \epsilon}{\epsilon+\rho_k}[
	\quad\text{for } k=1,2.
}
	and this property has been observed experimentally for higher dimensions $d=2,3$. Using higher values for $\tau_k$ actually often improves the (linear) convergence rate. Figure~\ref{fig:speed} displays a typical example of convergence, and exemplifies the usefulness of using large values of $\tau_k$, which leads to a speed-up of a factor 6 with respect to the usual Sinkhorn's choice $\tau_k=\tfrac{\epsilon}{\epsilon+\rho_k}$.
\end{rem}

\begin{figure}\centering
\begin{tabular}{@{}c@{\hspace{1mm}}c@{}}
\includegraphics[width=.49\linewidth]{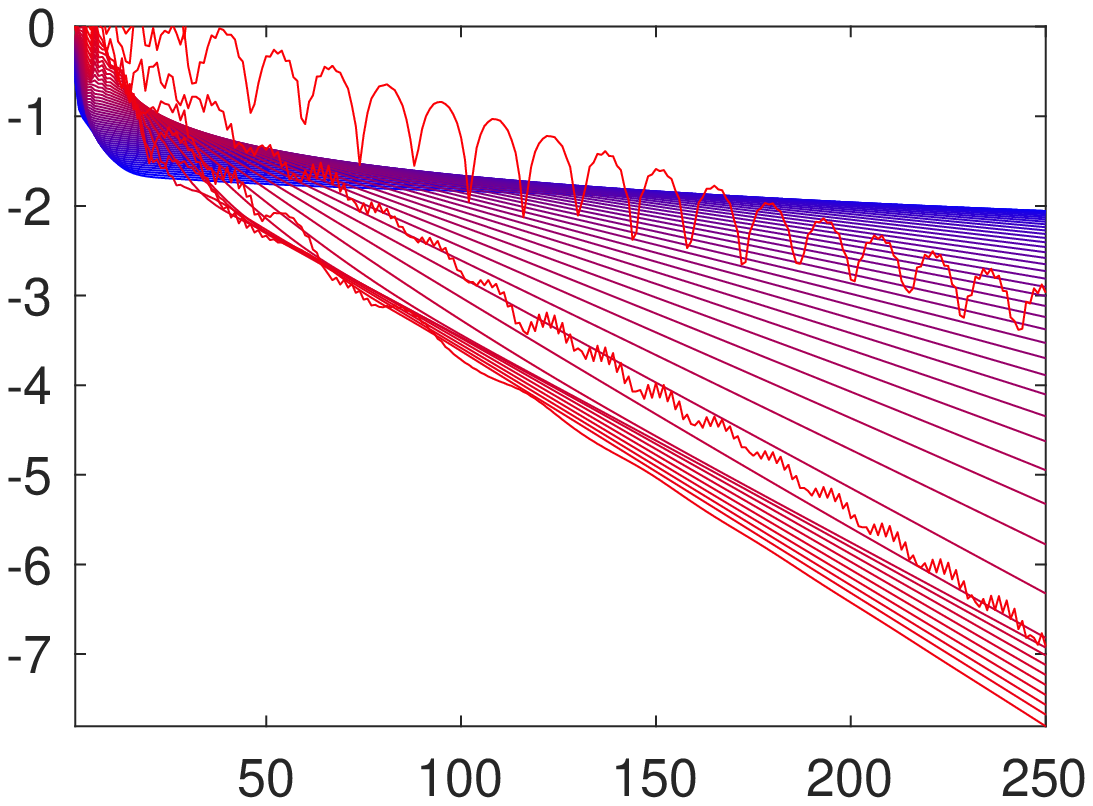}&
\includegraphics[width=.49\linewidth]{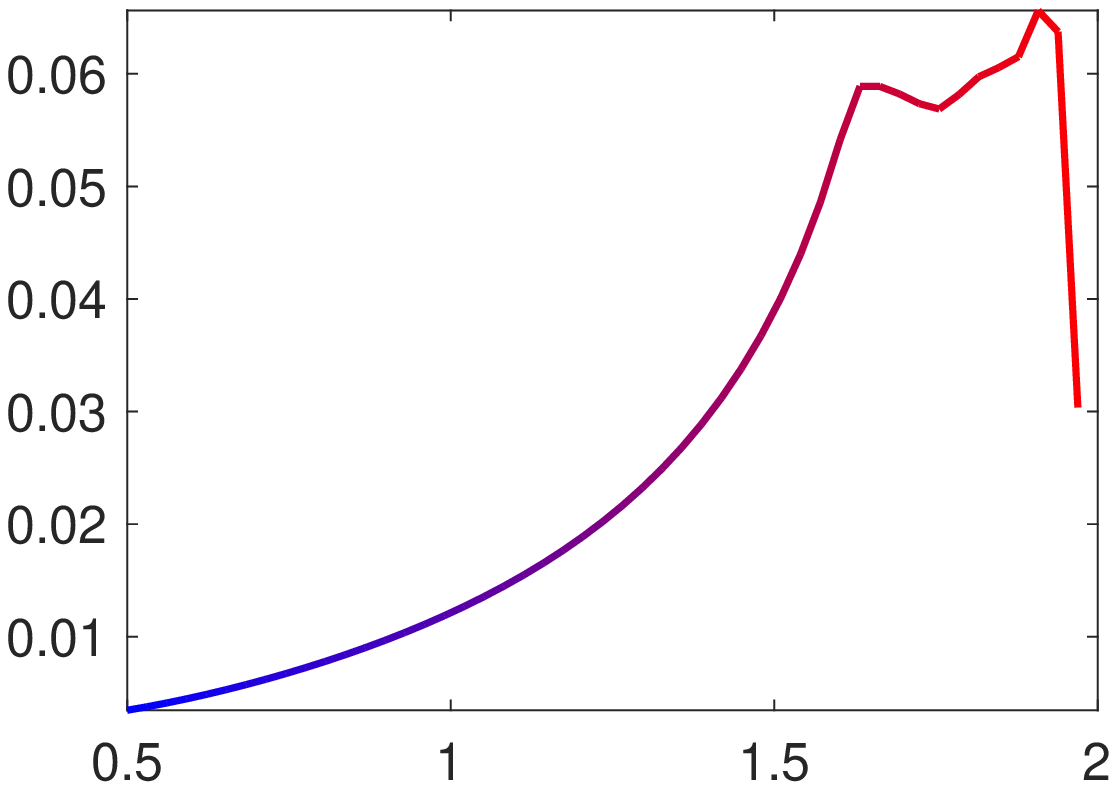}\\
$\log_{10}\norm{v^{(t+1)}-v^{(t)}}_\infty$ vs $t$ & 
Rate $\kappa$ vs. $\al$
\end{tabular}
\caption{Display of convergence of Sinkhorn Algorithm~\ref{alg:sinkhorn} for the example displayed on the first row of Figure~\ref{fig:intro}. 
Denoting $v^{(t)}$ the value of the variable $v$ at iteration $t$, the left plot shows the fixed point residual error for increasing values of $\tau_1=\tau_2=\tfrac{\al\epsilon}{\epsilon+\rho}$ with $\al \in [0.5,2]$ (blue to red). 
The algorithm exhibits a linear convergence rate, $\log_{10}\norm{v^{(t+1)}-v^{(t)}}_\infty \sim - \kappa t$ for some $\kappa>0$, and the right plot displays $\kappa$ as a function of $\al$.  
} \label{fig:speed}
\end{figure}

\begin{rem}[Heavy ball and acceleration]
	The extrapolation steps with weights $(\tau_1,\tau_2)$ are reminiscent of the heavy ball method~\cite{polyak1964some} to accelerate the convergence of first order optimization methods. 
	Non-linear acceleration methods would also be applicable, and would lift the need to manually tune the parameters $(\tau_1,\tau_2)$, see for instance~\cite{anderson1965iterative}.
	Similar acceleration technics can also be derived by using accelerate first order schemes (such as Nesterov's algorithm~\cite{nesterov1983method} or FISTA~\cite{beck2009fast}) directly on the dual problem~\eqref{eq-dual-pb}. Using second order information through quasi-Newton (L-BFGS) or conjugate-gradient is another option. 
	We leave for future work the exploration of these variants of our basic iterates. 
\end{rem}

\begin{rem}[Stability]In contrast to the usual implementation of Sinkhorn's algorithm, which is numerically unstable for small $\epsilon$ because it requires to compute $e^{u/\epsilon}$ and $e^{v/\epsilon}$, the proposed iterations using the LSE operator are stable. The algorithm can thus be run for arbitrary small $\epsilon$, although the linear speed of convergence is of course impacted.   
\end{rem}

\begin{algorithm}[t]
\fbox{\hspace{-.1in}\parbox{\columnwidth}{%
\begin{algorithmic}
\Function{Quantum-Sinkhorn}{$\mu,\nu,c,\epsilon,\rho_1,\rho_2$}
	\algspace
	\State $\foralls k=1,2, \quad \tau_k \in ]0,\tfrac{2 \epsilon}{\epsilon+\rho_k}[$, 
	\Let{$\foralls (i,j) \in I \times J, \quad (u_i,v_j)$}{$(0_{d \times d}, 0_{d \times d})$}
	\For{$s=1,2,3,\ldots$}
		\Let{$K$}{$\Kern(u,v)$}
		\State $\foralls i \in I, \quad u_i \RelaxAssign{\tau_1} \LSE_j(K_{i,j})-\log(\mu_i)$
		\Let{$K$}{$\Kern(u,v)$}
		\State $\foralls j \in J, \quad v_j \RelaxAssign{\tau_2} \LSE_i(K_{i,j})-\log(\nu_j)$
	\EndFor
	\State\Return{$(\ga_{i,j} = \exp(K_{i,j}))_{i,j}$}
\EndFunction
  \end{algorithmic}
}}
\caption{Quantum-Sinkhorn iterations to compute the optimal coupling $\ga$ of the regularized transportation problem~\eqref{eq-Kantorovich-regul}. The operator $\Kern$ is defined in~\eqref{eq-defn-K}.\label{alg:sinkhorn}}
\end{algorithm}

\begin{rem}[log and exp computations]
A major computational workload of the Q-Sinkhorn Algorithm~\ref{alg:sinkhorn} is the repetitive computation of matrix exp and log. 
For $d \in \{2,3\}$ it is possible to use closed-form expressions to diagonalize the tensors, so that the overall complexity is comparable with the usual scalar case $d=1$.
While the applications in Section~\ref{sec:appli} only require these low-dimensional settings, high dimensional problems are of interest, typically for machine learning applications.  In these cases, one has to resort to iterative procedures, such as rapidly converging squaring schemes~\cite{HighamExp,HighamLog}.
\end{rem}

\begin{rem}[Computational complexity]
For low-dimensional problems (typically for those considered in Section~\ref{sec:appli}), the Q-Sinkhorn Algorithm~\ref{alg:sinkhorn} scales to grid sizes of roughly 5k points (with machine-precision solutions computed in a few minutes on a standard laptop).
For large scale grids, even storing the full coupling $\ga$ becomes prohibitive. We however observed numerically that, similarly  to the usual scalar case, the optimal $\ga$ solving~\eqref{eq-Kantorovich-regul} is highly sparse (up to machine precision for small enough $\epsilon$).
We thus found that the use of the multi-scale refinement strategy introduced in~\cite{Schmitzer2016} is able to make the Q-Sinkhorn scale to high resolution grids.  It is not used to produce the figures of this article, but it is available in the companion computational toolbox.
\end{rem}

\begin{rem}[Gurvits' non-commutative Sinkhorn]
Let us insist on the fact that the proposed Q-Sinkhorn Algorithm~\ref{alg:sinkhorn} is unrelated to Gurvits' Sinkhorn algorithm~\cite{gurvits2004classical}. While Gurvits' iterations compute a coupling between a pair of input tensors, our method rather couples two \textit{fields} of tensors (viewed as tensor-valued measures). Our usage of the wording  ``quantum'' refers to the notion of quantum entropy~\eqref{eq-h-quantum} and is not inspired by quantum physics.
\end{rem}

\subsection{Trace-Constrained Extension}

The quantum OT problem~\eqref{eq-Kantorovich} does not impose that the marginals of the coupling $\gamma$ match exactly the inputs $(\mu,\nu)$. It is only in the limit $(\rho_1,\rho_2) \rightarrow (+\infty,+\infty)$ that an exact match is obtained, but as explained in Section~\ref{sec-tensor-ot}, this might leads to an empty constraint set.

To address this potential issue, we propose to rather only impose the \textit{trace} of the marginals to match the trace of the input measures, in order to guarantee conservation of mass (as measured by the trace of the tensors). We thus propose to solve the entropy regularized problem~\eqref{eq-Kantorovich-regul} with the extra constraint
\eq{
	\foralls i \in I, \quad \sum_j \tr(\ga_{i,j}) = \tr(\mu_i) \qandq
	\foralls j \in J, \quad \sum_i \tr(\ga_{i,j}) = \tr(\nu_j).
}
These two extra constraints introduce two dual Lagrange multipliers $(\al,\be) \in \RR^I \times \RR^J$ and the optimal coupling relation~\eqref{eq-primal-dual-scaling} is replaced by 
\eq{\label{eq-primal-dual-scaling_bis}
		\foralls (i,j), \quad \ga_{i,j} = \exp\pa{ \Kern(u,v,\al,\be)_{i,j} }
}
\eq{
		\qwhereq
		\Kern(u,v,\al,\be)_{i,j} \eqdef -\frac{c_{i,j} + \rho_1 u_i + \rho_2 v_j + \al_i + \be_j}{\epsilon}.
}

Q-Sinkhorn algorithm~\ref{alg:sinkhorn} is extended to handle these two extra variables $(\al,\be)$ by simply adding two steps to update these variables
\begin{align*}
	\foralls i \in I, \quad \al_i \leftarrow \al_i + \epsilon \LSTE_j( K )_i \qwhereq K \eqdef \Kern(u,v,\al,\be), \\
	\foralls j \in J, \quad \be_j \leftarrow \be_j + \epsilon \LSTE_i( K )_j \qwhereq K \eqdef \Kern(u,v,\al,\be).
\end{align*}
where we introduced the log-sum-trace-exp operator
\eq{
	\LSTE_j( K ) \eqdef \Big( \log \sum_j \tr( \exp(K_{i,j}) ) \Big)_i
}
(and similarly for $\LSTE_i$). Note that in this expression, the $\exp$ is matrix-valued, whereas the $\log$ is real-valued.

\subsection{Numerical Illustrations}
\label{sec-numerics-interp}

Figures~\ref{fig:intro},~\ref{fig:1d-interp} and~\ref{fig:qot-vs-ot} illustrate on synthetic examples of input tensor fields $(\muA,\muB)$ the Q-OT interpolation method. 
We recall that it is obtained in two steps:
\begin{enumerate}
	\item One first computes the optimal $\ga$ solving~\eqref{eq-Kantorovich-regul} using Sinkhorn iterations (Algorithm~\ref{alg:sinkhorn}).
	\item Then, for any $t \in [0,1]$, one computes $\mu_t$ using this optimal $\ga$ with formula~\eqref{eq-interpolating}.
\end{enumerate}
   
Figure~\ref{fig:1d-interp} shows examples of interpolations on a 1-D domain $X=Y=[0,1]$ with tensors of dimension $d=2$ and $d=3$, and a ground cost $c_{i,j}=|x_i-y_j|^2\Id_{d \times d}$. It compares the OT interpolation, which achieves a ``mass displacement,'' to the usual linear interpolation $(1-t)\mu+t\nu$, which only performs a pointwise interpolation of the tensors. 

Figure~\ref{fig:qot-vs-ot} shows the effect of taking into account the anisotropy of tensors into the definition of OT. In the case of isotropic tensors (see Remark~\ref{rem-classical-ot}), the method reduces to the usual scalar OT, and in 1-D it corresponds to the monotone re-arrangement~\cite{santambrogio2015optimal}. In contrast, the Q-OT of anisotropic tensors is forced to reverse the ordering of the transport map in order for tensors with similar orientations to be matched together. 
This example illustrates that the behaviour of our tensor interpolation is radically different from only applying classical scalar-valued OT to the trace of the tensor (which would result in the same coupling as the one obtained with isotropic tensors, Figure~\ref{fig:qot-vs-ot}, left). 

Figure~\ref{fig:intro} shows larger scale examples. 
The first row corresponds to $X=Y=[0,1]^2$ and $d=2$, with cost $c_{i,j}=\norm{x_i-y_j}^2\Id_{2 \times 2}$, which is a typical setup for image processing.
The second row corresponds to $X=Y$ being a triangulated mesh of a surface, and the cost is proportional to the squared geodesic distance $c_{i,j}=d_X(x_i,y_j)^2 \Id_{2\times 2}$. 


\begin{figure}\centering
\begin{tabular}{@{}c@{}|@{}c@{}}
\includegraphics[width=.49\linewidth]{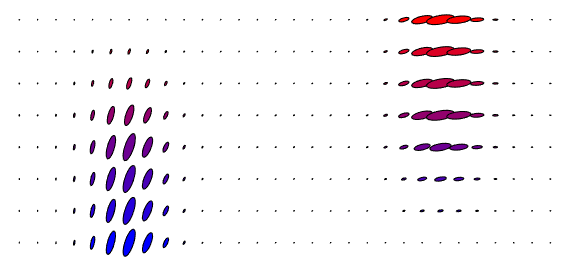}&
\includegraphics[width=.49\linewidth]{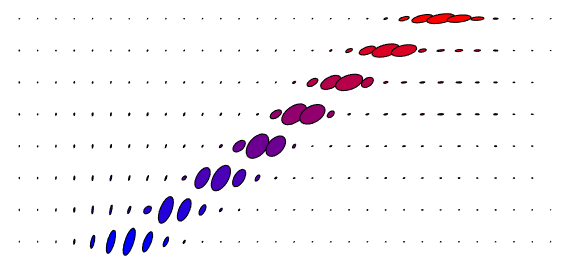}\\\hline
\includegraphics[width=.49\linewidth]{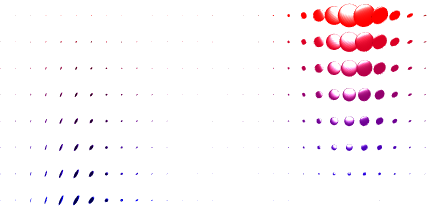}&
\includegraphics[width=.49\linewidth]{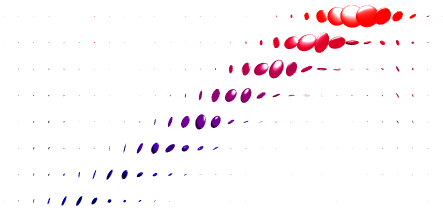}\\\hline
Linear interpolation & Quantum OT
\end{tabular}
\caption{Comparison of linear and quantum-OT interpolation (using formula~\eqref{eq-interpolating}). 
Each row shows a tensor field $\mu_t$ (top $d=2$, bottom $d=3$) along a linear segment from $t=0$ to $t=1$ ($t$ axis is vertical).
} \label{fig:1d-interp}
\end{figure}

\begin{figure}\centering
\begin{tabular}{@{}c@{}|@{}c@{}}
\includegraphics[width=.49\linewidth,trim=40 55 30 48,clip]{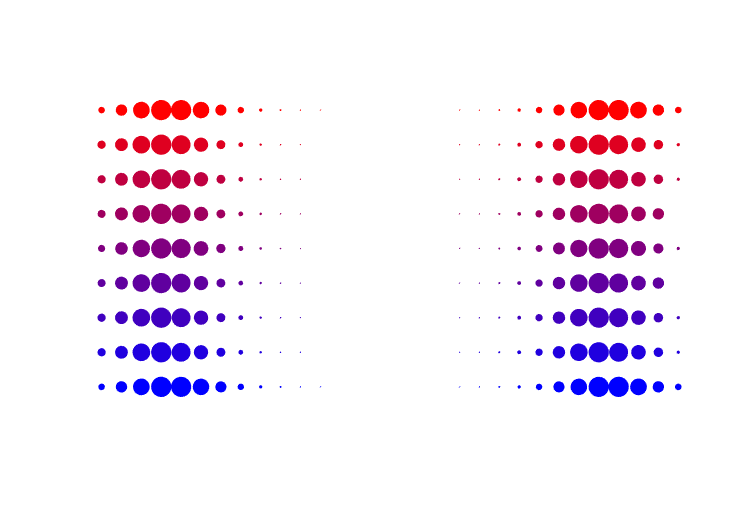}&
\includegraphics[width=.49\linewidth,trim=40 55 30 48,clip]{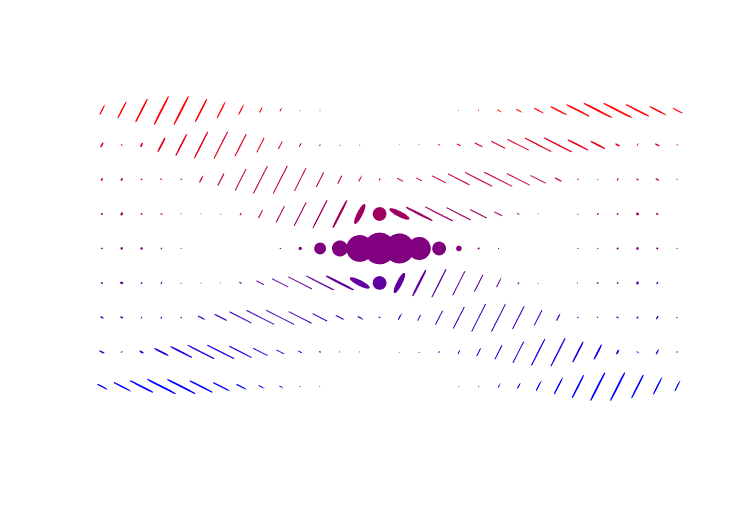}\\\hline
Classical OT & Quantum OT
\end{tabular}
\caption{Comparison of classical OT (i.e. between isotropic tensors) and quantum-OT (between anisotropic tensors) interpolation (using formula~\eqref{eq-interpolating}), using the same display as Figure~\ref{fig:1d-interp}. 
} \label{fig:qot-vs-ot}
\end{figure}


\section{Quantum Barycenters}
\label{sec-q-bary}

Following Agueh and Carlier~\shortcite{agueh-2011} (see also~\cite{benamou-2015,solomon-2015} for numerical methods using entropic regularization), we now propose a generalization of the OT problem~\eqref{eq-Kantorovich}, where, instead of coupling only two input measures, one tries to couple an arbitrary set of inputs, and compute their Fr\'echet means. 

\newcommand{\weight}{w}

\subsection{Barycenter Optimization Problem}

Given some input measures $(\mu^\ell)_\ell$, the quantum barycenter problem reads
\eql{\label{eq-defn-barycenters}
	\umin{\nu} \sum_\ell \weight_\ell W_\epsilon(\mu^\ell,\nu), 
}
where $(\weight_\ell)_\ell$ is a set of positive weights normalized so that $\sum_\ell \weight_\ell=1$. In the following, for simplicity, we set
\eq{
	\rho_1 = \rho 
	\qandq
	\rho_2 = +\infty
}
in the definition~\eqref{eq-Kantorovich} of $W_\epsilon$. Note that the choice $\rho_2=+\infty$ corresponds to imposing the exact hard marginal constraint $\gamma^\top \ones_J=\nu$. 

\begin{rem}[Barycenters between single Dirac masses]
	If all the input measures are concentrated on single Diracs $\mu^\ell=P_\ell \de_{x_\ell}$, then the single Dirac barycenter (unregularized, i.e., $\epsilon=0$) for a cost $d_X(x,y)^\al \Id_{d \times d}$ is $P \de_x^\star$ where $x^\star \in X$ is the usual barycenter for the distance $d_X$, solving 
	\eq{
		x^\star \in \argmin_{x} 
			\Ee(x) = \sum_\ell \weight_\ell d_X^\al(x_\ell,x) 
	}
	and the barycentric matrix is
	\eql{\label{eq-barycentric-matrix}
		P = e^{-\frac{\Ee(x^\star)}{\rho}} \exp\Big(\sum_{\ell} \weight_\ell \log(P_\ell)\Big).
	}
	Figure~\ref{fig:interp} illustrates the effect of a pointwise interpolation (i.e. at the same location $x_\ell$ for all $\ell$) between tensors.
\end{rem}

\begin{figure}\centering
\fbox{\includegraphics[width=.46\linewidth]{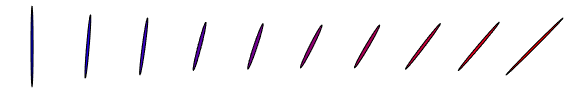}}
\fbox{\includegraphics[width=.46\linewidth]{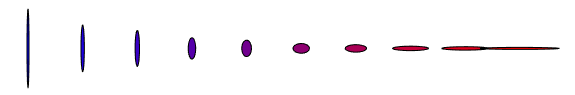}}
\caption{Two examples of pointwise (without transportation) interpolations, using formula~\eqref{eq-barycentric-matrix}. Here $P_1$ and $P_2$ are represented using the blue/red ellipses on the left/right, and weights are $(w_1,w_2)=(1-t,t)$ for $t \in [0,1]$ from left to right.} \label{fig:interp}
\end{figure}

\newcommand{\BaryImg}[2]{{\includegraphics[width=.095\linewidth,trim=18 160 18 160,clip]{2d-bary/barycenter-#1-#2}}}
\newcommand{\BaryImgLine}[1]{
\BaryImg{#1}{1}&\BaryImg{#1}{2}&\BaryImg{#1}{3}&\BaryImg{#1}{4}&\BaryImg{#1}{5} 
}
\newcommand{\BarySurf}[2]{{\includegraphics[width=.095\linewidth,trim=25 10 25 10,clip]{mesh-bary/barycenter-#1-#2}}}
\newcommand{\BarySurfLine}[1]{
\BarySurf{#1}{1}&\BarySurf{#1}{2}&\BarySurf{#1}{3}&\BarySurf{#1}{4}&\BarySurf{#1}{5} 
}

\begin{figure*}\centering
\HighResFig{
\begin{tabular}{@{}c@{}c@{}c@{}c@{}c@{}}
\BaryImgLine{1}\\
\BaryImgLine{2}\\
\BaryImgLine{3}\\
\BaryImgLine{4}\\
\BaryImgLine{5}
\end{tabular}
\hspace{1mm}
\begin{tabular}{@{}c@{}c@{}c@{}c@{}c@{}}
\BarySurfLine{1}\\
\BarySurfLine{2}\\
\BarySurfLine{3}\\
\BarySurfLine{4}\\
\BarySurfLine{5}
\end{tabular}
}{\includegraphics[width=\linewidth]{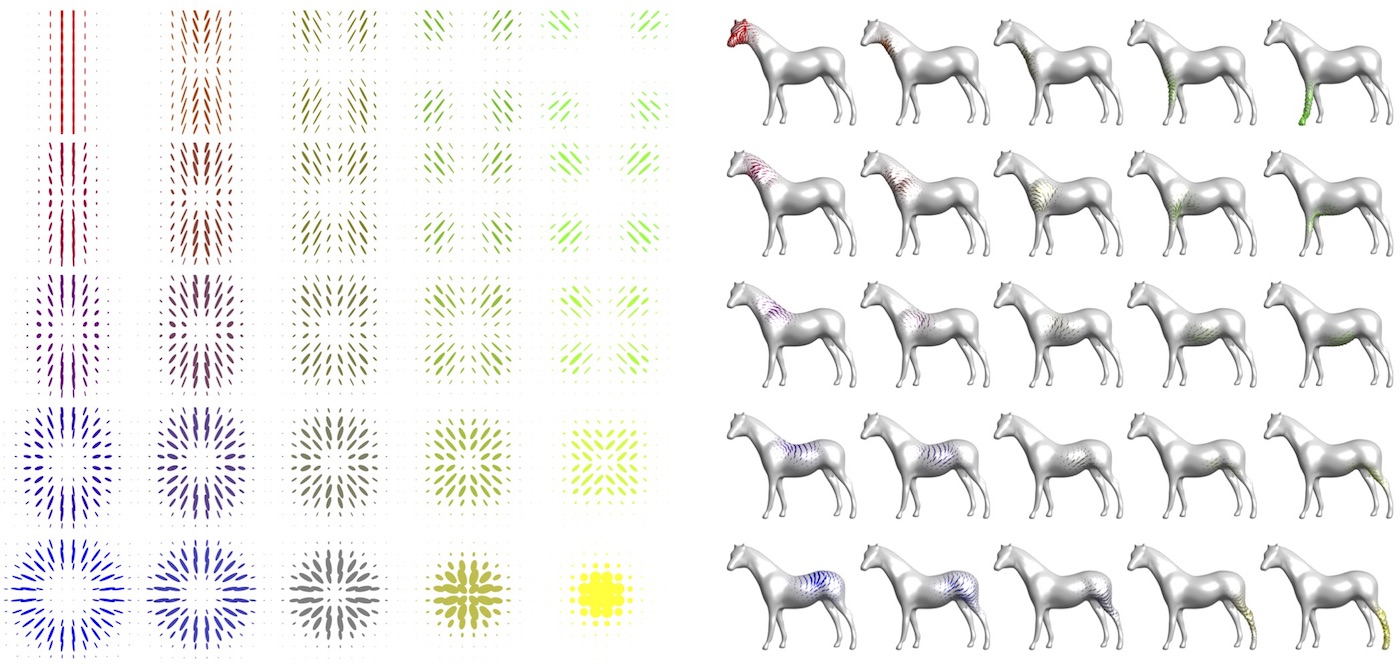}}
\caption{$5 \times 5$ barycenters of four input measures (displayed in the four corners). The weighs $w \in \RR^4$ correspond to bilinear interpolation weights~\eqref{eq-bilinear} inside the square.
} \label{fig:barycenters}
\end{figure*}

Problem~\eqref{eq-defn-barycenters} is convex, and similarly to~\eqref{eq-dual-pb}, it can be rewritten in dual form.

\begin{prop}
The optimal $\nu$ solving~\eqref{eq-defn-barycenters} is the solution of
\begin{multline}\label{eq-dual-bary}		
		\umax{(u^\ell,v^\ell)} \umin{\nu}
				- 
				\sum_\ell w_\ell 
					\tr\Big[
						\rho \sum_i e^{u_i^\ell+\log(\mu_i^\ell)} 
					+    \sum_j \nu_j v_j^\ell
					+    \epsilon \sum_{i,j}  e^{ \Kern(u^\ell,v^\ell)_{i,j} }
			 \Big], 
\end{multline}	
where here we define $\Kern$ as
\eql{\label{eq-def-k-bary}
	\Kern(u,v)_{i,j} \eqdef -\frac{c_{i,j} + \rho u_i +  v_j}{\epsilon}.
}
\end{prop}

\begin{algorithm}[t]
\fbox{\hspace{-.1in}\parbox{\columnwidth}{%
\begin{algorithmic}
\Function{Quantum-Barycenter}{$(\mu_\ell)_{\ell=1}^L,c,\epsilon,\rho$}
	\algspace
	\State Choose $\tau_1 \in ]0,\tfrac{2 \epsilon}{\epsilon+\rho}[$, 
		$\tau_2 \in ]0,2[$.
	\Let{$\foralls (i,j) \in I \times J, \quad (u_i,v_j)$}{$(0_{d \times d}, 0_{d \times d})$}
	\For{$s=1,2,3,\ldots$.}
		\For{$\ell=1,\ldots,L$}
			\Let{$K^\ell$}{$\Kern(u^\ell,v^\ell)$, }
			\State $\foralls i \in I, \quad u_i^\ell \RelaxAssign{\tau_1} \LSE_j(K_{i,j}^\ell)-\log(\mu_i^\ell)$, 
			\Let{$K^\ell$}{$\Kern(u^\ell,v^\ell)$.}
		\EndFor
		\State $\foralls j \in J, \quad \log(\nu_j) \leftarrow 
			\sum_\ell w_\ell ( \LSE_i(K_{i,j}^\ell) + v^\ell_j/\epsilon ).$
		\For{$\ell=1,\ldots,L$}
			\State $\foralls j \in J, \quad v_j^\ell \RelaxAssign{\tau_2} \epsilon \LSE_i(K_{i,j}^\ell) + v^\ell_j - \epsilon \log(\nu_j).$	
		\EndFor		
	\EndFor
	\State\Return{$\nu$}
\EndFunction
  \end{algorithmic}
}}
\caption{Quantum-Barycenter iterations to compute the optimal barycenter measure $\nu$ solving~\eqref{eq-defn-barycenters}. The operator $\Kern$ is defined in~\eqref{eq-def-k-bary}. \label{alg:barycenter}}
\end{algorithm}

\subsection{Quantum Barycenter Sinkhorn}

Similarly to Proposition~\ref{prop-fixed-points}, the dual solutions of~\eqref{eq-dual-bary} satisfy a set of coupled fixed point equations:

\begin{prop}
Optimal $(u^\ell,v^\ell)_\ell$ for~\eqref{eq-dual-bary} satisfy 
\begin{align}
	\foralls (i,\ell), \quad \LSE_j(\Kern(u^\ell,v^\ell)_{i,j})-\log(\mu_i^\ell), \label{eq-fixed-point-u-bary} 
		&= u_i^\ell \\
	\foralls (j,\ell), \quad \LSE_i(\Kern(u^\ell,v^\ell)_{i,j}) \label{eq-fixed-point-v-bary}
		&= \log(\nu_j)\\ 
		 \textstyle \sum_\ell \weight_\ell v^\ell &= 0. \label{eq-fixed-point-nu-bary}
\end{align}
\end{prop}
\begin{proof}
The proof of~\eqref{eq-fixed-point-u-bary} and~\eqref{eq-fixed-point-v-bary} is the same as the one of Proposition~\ref{prop-fixed-points}.
Minimization of~\eqref{eq-dual-bary} on $\nu$ leads to~\eqref{eq-fixed-point-nu-bary}. 
\end{proof}

The extension of the quantum Sinkhorn algorithm to solve the barycenter problem~\eqref{alg:barycenter} is detailed in Algorithm~\ref{alg:barycenter}. It alternates between the updates of $u$, $\nu$ and $v$, using the relaxed version of the fixed point equations~\eqref{eq-fixed-point-u-bary}, \eqref{eq-fixed-point-v-bary} and~\eqref{eq-fixed-point-nu-bary}. The notation $\RelaxAssign{\tau}$ refers to a relaxed assignment as defined in~\eqref{eq-dfn-relaxed-assign}. 

\begin{rem}[Choice of $\tau$]
Remark~\ref{rem-choice-tau} also applies for this Sinkhorn-like scheme, and setting $(\tau_1,\tau_2)=(\tfrac{\epsilon}{\rho+\epsilon},1)$ leads, in the scalar case $d=1$, to the algorithm in~\cite{2016-chizat-sinkhorn}. We found experimentally that this choice leads to contracting (and hence linearly converging) iterations, and that higher values of $\tau$ usually accelerate the convergence rate. 
\end{rem}

\begin{rem}[Scalar and isotropic cases]
Note that in the scalar case $d=1$ and for isotropic input tensors (multiples of the identity), one retrieves the provably convergent unbalanced barycenter algorithm in~\cite{2016-chizat-sinkhorn}.
\end{rem}

\subsection{Numerical Illustrations}

Figure~\ref{fig:barycenters} shows examples of barycenters $\nu$ solving~\eqref{eq-defn-barycenters} between four input measures $(\mu^\ell)_{\ell=1}^4$. The horizontal/vertical axes of the figures are indexed by $(t_1,t_2) \in [0,1]^2$ (on a $5 \times 5$ grid) and parameterize the weights $(w_\ell)_{\ell=1}^4$ appearing in~\eqref{eq-defn-barycenters} as
\eql{\label{eq-bilinear}
	(w_1,w_2,w_3,w_4) \eqdef ( (1-t_1)(1-t_2), (1-t_1)t_2, t_1(1-t_2), t_1,t_2 ). 
}
The left part of Figure~\ref{fig:barycenters} corresponds to measures on $X=Y=[0,1]^2$ with $d=2$ and ground cost $c_{i,j}=\norm{x_i-x_j}^2\Id_{2 \times 2}$.
The right part of Figure~\ref{fig:barycenters} corresponds to measures on $X=Y$ being a surface mesh with $d=2$ (the tensors are defined on the tangent planes) and a ground cost is $c_{i,j}=d_X(x_i,x_j)^2\Id_{2 \times 2}$ where $d_X$ is the geodesic distance on the mesh.


\section{Applications}
\label{sec:appli}

This section showcases four different applications of Q-OT to register and interpolate tensor fields.
Unless otherwise stated, the data is normalized to the unit cube $[0,1]^d$ (here $d=2$ for images) and discretized on grids of $|I|=|J|=50^d$ points. 
The regularization parameter is set to $\epsilon=0.08^2$, the fidelity penalty to $\rho=1$, and the relaxation parameter for Sinkhorn to $\tau_k=\tfrac{1.8 \epsilon}{\epsilon+\rho_k}$. 


\subsection{Anisotropic Space-Varying Procedural Noise}

\newcommand{\BumpFig}[1]{\includegraphics[width=.17\linewidth,trim=140 10 125 0,clip]{mesh-bump/anisodiffus-#1}}
\newcommand{\TextureImg}[2]{\includegraphics[width=.19\linewidth]{textures/#1/interpol-#2}}
\begin{figure}\centering
\begin{tabular}{@{}c@{\hspace{.5mm}}c@{\hspace{.5mm}}c@{\hspace{.5mm}}c@{\hspace{.5mm}}c@{\hspace{.5mm}}c@{}}
\TextureImg{2d-bump-donut}{render-1}&
\TextureImg{2d-bump-donut}{render-3}&
\TextureImg{2d-bump-donut}{render-5}&
\TextureImg{2d-bump-donut}{render-7}&
\TextureImg{2d-bump-donut}{render-9}&
\includegraphics[height=.19\linewidth]{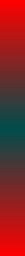} \\
\BumpFig{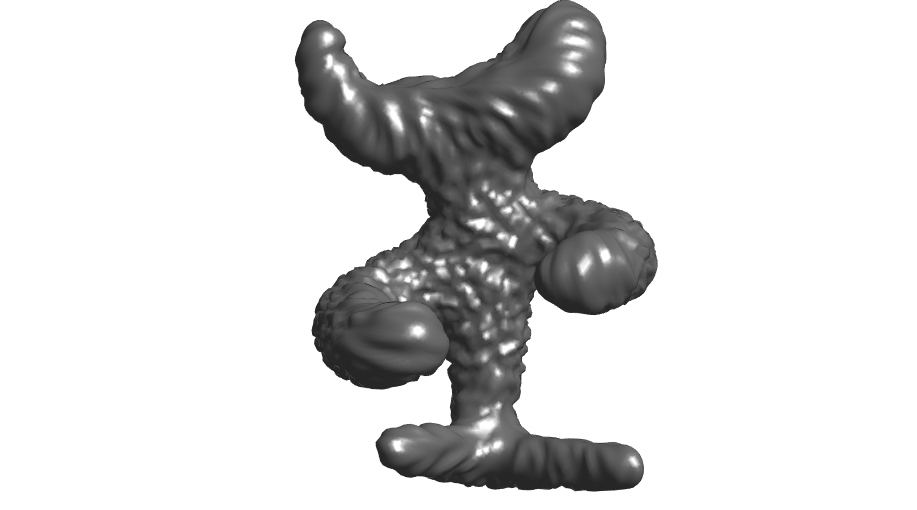}&
\BumpFig{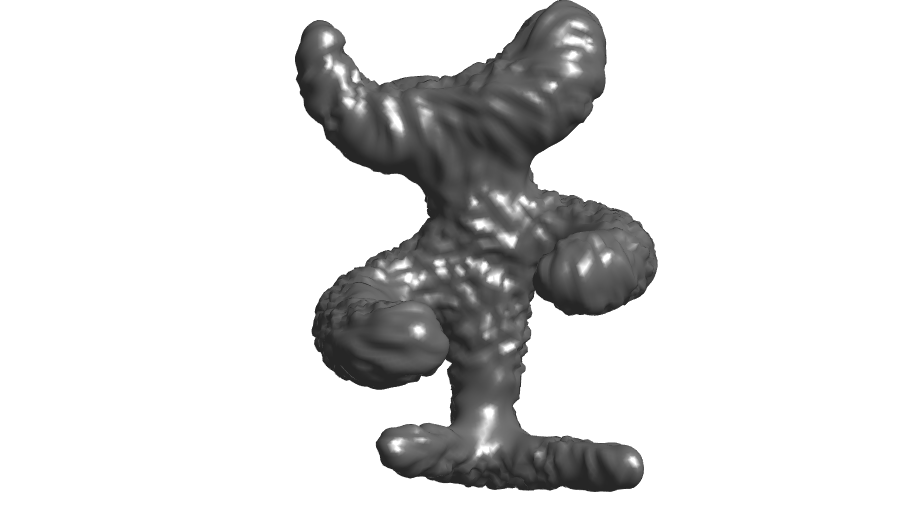}&
\BumpFig{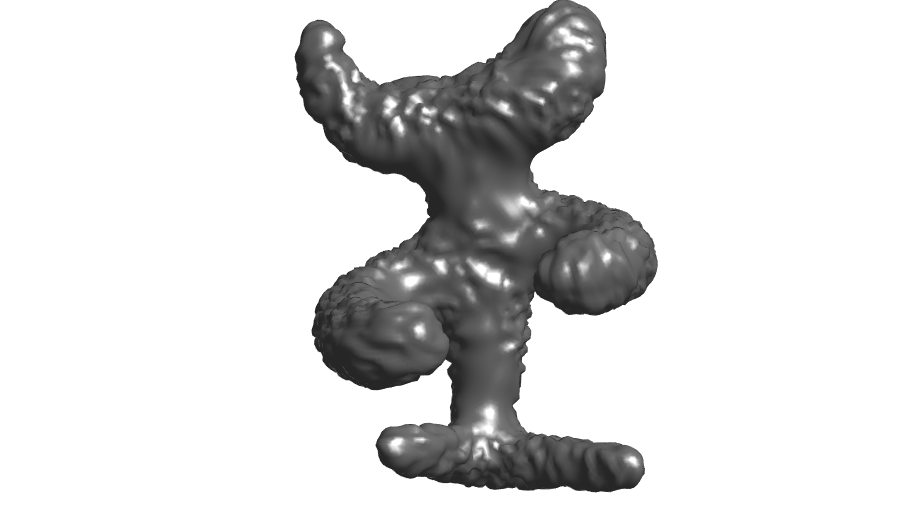}&
\BumpFig{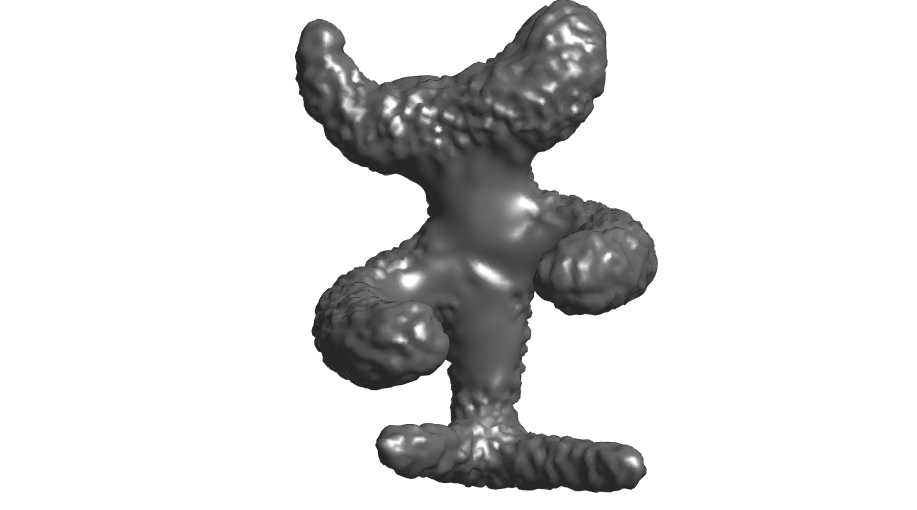}&
\BumpFig{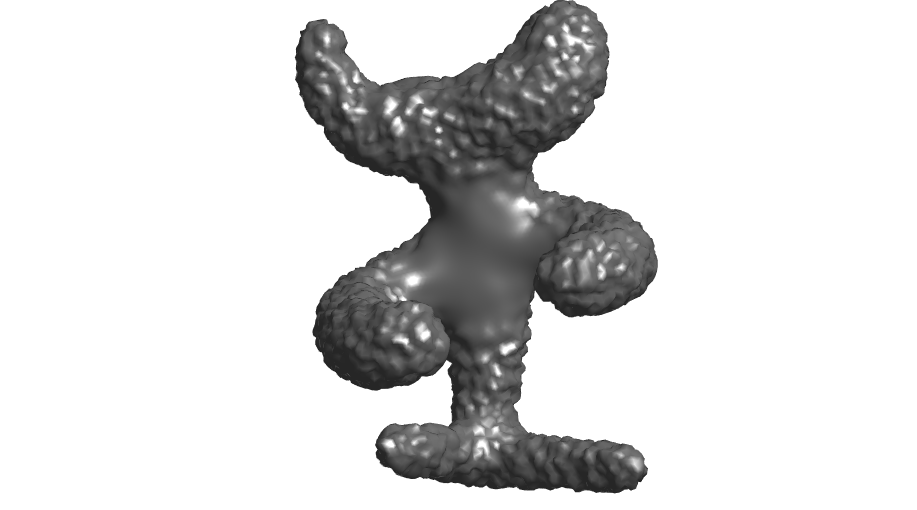}&\\
$t=0$ & $t=1/4$ & $t=1/2$ & $t=3/4$ & $t=1$
\end{tabular}
\caption{Example of interpolation between two input procedural anisotropic noise functions. The PSD tensor field parameterizing the texture are displayed on Figure~\ref{fig:intro}. The colormap used to render the anisotropic texture is displayed on the last column.  
} \label{fig:texture}
\end{figure}

Texture synthesis using procedural noise functions is widely used in rendering pipelines and video games because of both its low storage cost and the fact that it is typically parameterized by a few meaningful parameters~\cite{LagaeSurvey}. 
Following Lagae et al.~\shortcite{LagaImproving} we consider here a spatially-varying Gabor noise function (i.e.\ non-stationary Gaussian noise), whose covariance function is parameterized using a PSD-valued field $\mu$. 
Quantum optimal transport allows to interpolate and navigate between these noise functions by transporting the corresponding tensor fields. 
The initial Gabor noise method makes use of sparse Gabor splattering~\cite{LagaeSurvey} (which enables synthesis at arbitrary resolution and zooming). For simplicity, we rather consider here a more straightforward method, where the texture $f_{t_0}$ is obtained by stopping at time $t=t_0$ an anisotropic diffusion guided by the tensor field $\mu$  of a high frequency noise $\Nn$ (numerically a white noise on a grid)
\eq{
	\frac{\partial_t f_t}{\partial t} = \text{div}( \mu \nabla f_t ), \qwhereq
	f_{t=0} \sim \Nn, 
}
where $(\mu \nabla f_t)(x) \eqdef \mu(x) (\nabla f_t(x))$ is the vector field obtained by applying the tensor $\mu(x) \in \Ss_2^+$ to the gradient vector $\nabla f_t(x) \in \RR^2$. 
Locally around $x$, the texture is stretched in the direction of the main eigenvector of $\mu(x)$,  highly anisotropic tensor giving rise to elongated ``stripes'' as opposed to isotropic tensor generating ``spots.''

Numerically, $f$ is discretized on a 2-D grid, and $\mu$ is represented on this grid as a sum of Dirac masses~\eqref{eq-input-measures}. On Euclidean domains $X$, $\nabla$ and div are computed using finite differences, while on triangulated mesh, they are implemented using standard piecewise-linear finite element primitives. 
Figure~\ref{fig:texture} shows two illustrations of this method. The top row generates an animated color texture by indexing a non-linear black-red colormap (displayed on the right) using $f_t$. The bottom row generates an animated bump-mapped surface using $f_t$ to offset the mesh surface in the normal direction.


\subsection{Anisotropic Meshing}

Approximation with anisotropic piecewise linear finite elements on a triangulated mesh is a fundamental tool to address tasks such as discretizing partial differential equations, performing surface remeshing~\cite{alliez2003anisotropic} and image compression~\cite{demaret2006image}.
A common practice is to generate triangulations complying with a PSD tensor sizing field $\mu$, i.e. such that a triangle centered at $x \in X$ should be inscribed in the ellipsoid $\enscond{u \in X}{(u-x)^\top \mu(x) (u-x) \leq \de }$ for some $\de$ controlling the triangulation density. 
A well-known result is that, to locally approximate a smooth convex $C^2$ function $f$,  the optimal shapes of triangles is dictated by the Hessian $H f$ of the function (see~\cite{shewchuk2002good}). In practice, people use $\mu(x) = |H f(x)|^\al$ for some exponent $\al > 0$ (which is related to the quality measure of the approximation), where $|\cdot|^\al$ indicates the spectral application of the exponentiation (as for matrix exp or log).

Figure~\ref{fig:meshing} shows that Q-OT can be used (using formula~\eqref{eq-interpolating}) to interpolate between two sizing fields $(\mu,\nu)$, which are computed from the Hessians (with here $\al=1$) of two initial input images $(f,g)$.
The resulting anisotropic triangulations are computed using the method detailed in~\cite{peyre-iccv-09}. They corresponds to geodesic Delaunay triangulations for the Riemannian metric defined by the tensor field. 
This interpolation could typically be used to track the evolution of the solution of some PDE. 

\newcommand{\MeshingImg}[2]{\includegraphics[width=.195\linewidth]{meshing/#1/input-#2}}
\begin{figure}
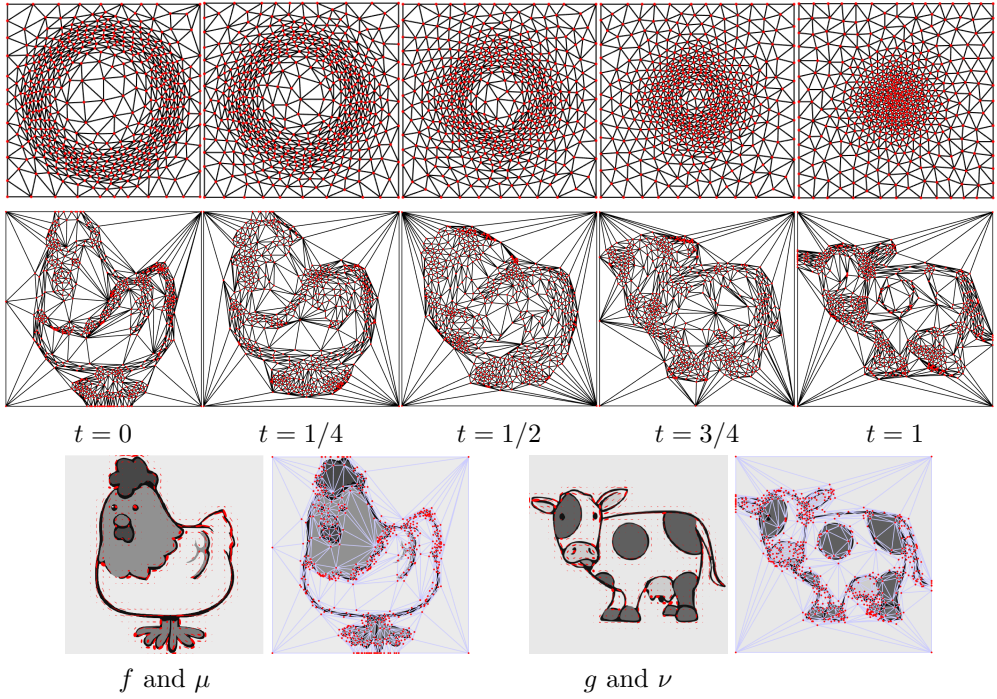
\centering
\begin{tabular}{@{}c@{}c@{}c@{}c@{}c@{}}
\MeshingImg{2d-bump-donut}{mesh-1}&
\MeshingImg{2d-bump-donut}{mesh-3}&
\MeshingImg{2d-bump-donut}{mesh-5}&
\MeshingImg{2d-bump-donut}{mesh-7}&
\MeshingImg{2d-bump-donut}{mesh-9}\\
\MeshingImg{images}{mesh-1}&
\MeshingImg{images}{mesh-3}&
\MeshingImg{images}{mesh-5}&
\MeshingImg{images}{mesh-7}&
\MeshingImg{images}{mesh-9}\\
$t=0$ & $t=1/4$ & $t=1/2$ & $t=3/4$ & $t=1$
\end{tabular}
\begin{tabular}{@{}c@{\hspace{1mm}}c@{\hspace{8mm}}c@{\hspace{1mm}}c@{}}
\MeshingImg{images}{images-1}&
\MeshingImg{images}{mesh-1-img}&
\MeshingImg{images}{images-2}&
\MeshingImg{images}{mesh-9-img} \\
$f$ and $\mu$ & & $g$ and $\nu$ &
\end{tabular}
\caption{Two examples of interpolation between two input sizing fields $(\mu_{t=0},\mu_{t=1})=(\mu,\nu)$. 
\textbf{First row:} triangulation evolution for the sizing fields displayed on Figure~\ref{fig:intro}.
\textbf{Second row:} the input sizing fields $(\mu_{t=0},\mu_{t=1})=(\mu,\nu)$ are displayed on the third row, and are defined using the absolute value ($\al=1$) of the Hessian of the underlying images $(f,g)$.
} \label{fig:meshing}
\end{figure}


\subsection{Diffusion Tensor Imaging}

Diffusion tensor magnetic resonance imaging (DTI) is a popular technique to image the white matter of the brain (see~\cite{wandell2016clarifying} for a recent overview). DTI measures the diffusion of water molecules, which can be compactly encoded using a PSD tensor field $\mu(x) \in \Ss_+^3$, whose anisotropy and size matches the local diffusivity. 
A typical goal of this imaging technique is to map the brain anatomical connectivity, and in particular track the  white matter fibers. This requires a careful handling of the tensor's energy (its trace) and anisotropy, so using Q-OT is a perfect fit for such data.

Figure~\ref{fig:dti} shows an application of Q-OT for the interpolation (using~\ref{eq-interpolating}) between 2-D slices from DTI tensor fields $(\mu,\nu)$ acquired on two different subjects. This data is extracted from the studies~\cite{pestilli2014evaluation,takemura2016ensemble}. These two patients exhibit different anatomical connectivity geometries, and Q-OT is able to track the variation in both orientation and magnitude of the diffusion tensors. This figure also compares the different data fidelity parameters $\rho \in \{0.05,1\}$. Selecting $\rho=1$ enforces an overly-strong conservation constraint and leads to interpolation artifacts (in particular some structure are split during the interpolation). In contrast, selecting $\rho=0.05$ introduces enough mass creation/destruction during the interpolation to be able to cope with strong inter-subject variability.

\newcommand{\DTIimg}[1]{\includegraphics[width=.195\linewidth,trim=85 20 85 20,clip]{dti/#1}}

\begin{figure}\centering
\HighResFig{
\begin{tabular}{@{}c@{\hspace{1mm}}c@{\hspace{1mm}}c@{\hspace{1mm}}c@{\hspace{1mm}}c@{}}
\DTIimg{one-two/interpol-rho1-1}&
\DTIimg{one-two/interpol-rho1-3}&
\DTIimg{one-two/interpol-rho1-5}&
\DTIimg{one-two/interpol-rho1-7}&
\DTIimg{one-two/interpol-rho1-9}\\
\DTIimg{one-two/interpol-rho005-1}&
\DTIimg{one-two/interpol-rho005-3}&
\DTIimg{one-two/interpol-rho005-5}&
\DTIimg{one-two/interpol-rho005-7}&
\DTIimg{one-two/interpol-rho005-9}\\
\DTIimg{four-two/interpol-rho005-1}&
\DTIimg{four-two/interpol-rho005-3}&
\DTIimg{four-two/interpol-rho005-5}&
\DTIimg{four-two/interpol-rho005-7}&
\DTIimg{four-two/interpol-rho005-9}\\
$t=0$ & $t=1/4$ & $t=1/2$ & $t=3/4$ & $t=1$
\end{tabular}
\begin{tabular}{@{}c@{\hspace{2mm}}c@{\hspace{2mm}}c@{}}
\includegraphics[width=.3\linewidth]{dti/one-two/original-trace-1-roi}&
\includegraphics[width=.3\linewidth]{dti/one-two/original-trace-2-roi}&
\includegraphics[width=.3\linewidth]{dti/four-two/original-trace-1-roi}\\
Subject $A$ & Subject $B$ & Subject $C$  
\end{tabular}
}{\includegraphics[width=\linewidth]{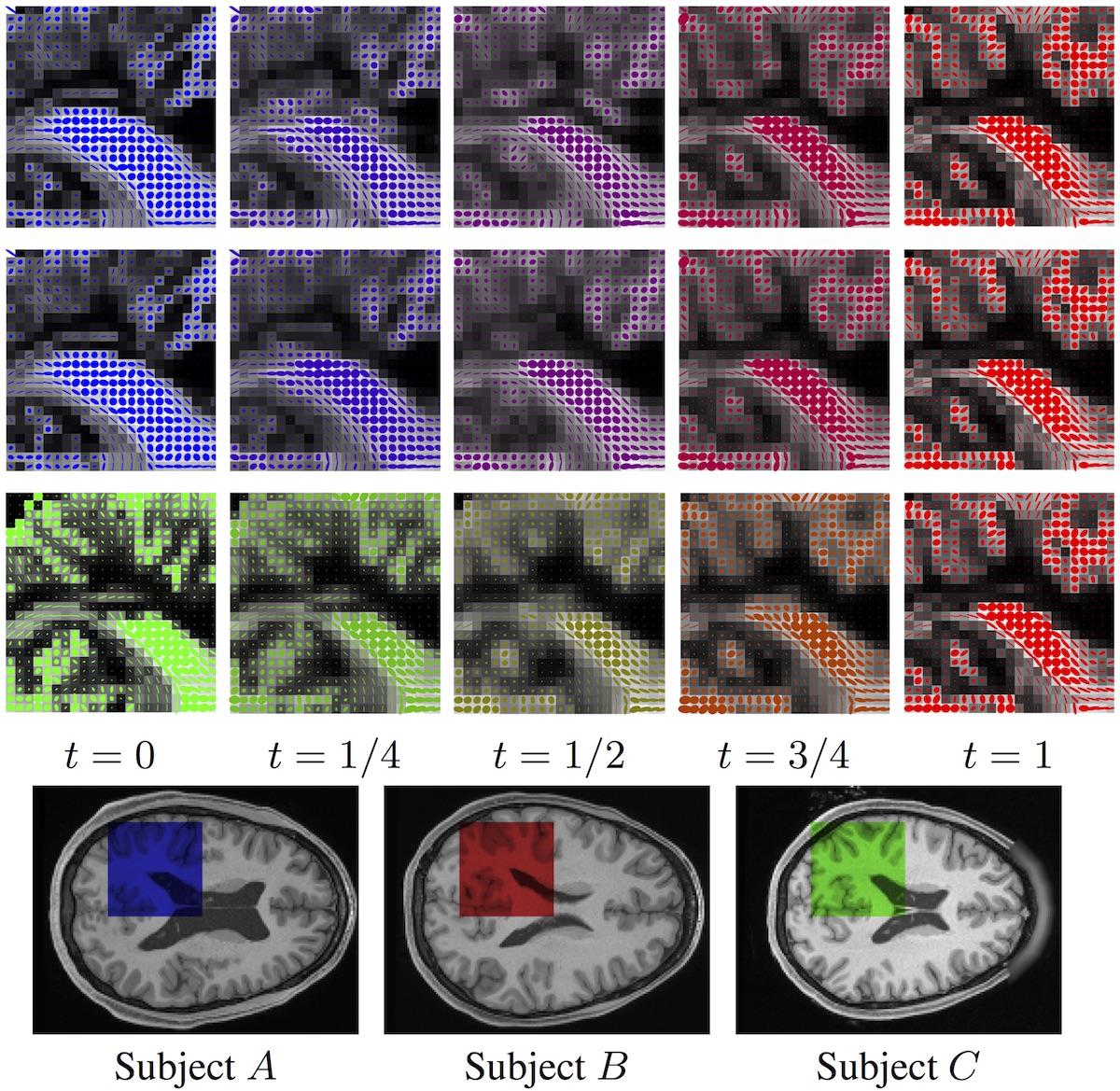}}
\caption{Interpolation between two 2-D slices of 3-D DTI tensor fields $(\mu,\nu)=(\mu_{t=0},\mu_{t=1})$. For readability, only the X/Y components of the tensors are displayed. 
\textbf{First row:} interpolation between subjects $(A,B)$ obtained using $\rho=1$. 
\textbf{Second row:} interpolation between subjects $(A,B)$ obtained using $\rho=0.05$. 
\textbf{Third row:} interpolation between subjects $(C,B)$ obtained using $\rho=0.05$. 
\textbf{Fourth row:} anatomical MRI images of subjects $(A,B,C)$ indicating the region of interest where the computations are performed. 
} \label{fig:dti}
\end{figure}

\subsection{Spectral Color Texture Synthesis}

As advocated initially in~\cite{galerne2011random}, a specific class of textured images (so-called micro-textures) is well-modeled using stationary Gaussian fields. In the following, we denote $p$ the pixel positions and $x$ the Fourier frequency indices. For color images, these fields are fully characterized by their mean $m \in \RR^3$ and their Fourier power spectrum, which is a tensor valued field $\mu(x)$ where, for each frequency $x$ (defined on a 2-D grid) $\mu(x) \in \CC^{3 \times 3}$ is a complex positive semi-definite hermitian matrix. 

In practice, $\mu(x)$ is estimated from an exemplar color image $f(p) \in \RR^3$ using an empirical spectrogram 
\eql{\label{eq-power-spectrum}
	\mu(x) \eqdef \frac{1}{K} \sum_{k=1}^K \hat f_k(x) \hat f_k(x)^* \in \CC^{3 \times 3}
}
where $\hat f_k$ is the Fourier transform of $f_k(p) \eqdef f(p) w_k(p)$ (computed using the FFT), $w_k$ are windowing functions centred around $K$ locations in the image plane, and $v^* \in \CC^{1 \times 3}$ denotes the transpose-conjugate of a vector $v \in \CC^{3 \times 1}$. 
Increasing the number $K$ of windowed estimations helps to avoid having rank-deficient covariances ($K=1$ leads to a field $\mu$ of rank-1 tensors).

Randomized new textures are then created by generating random samples $F(p) \in \RR^3$ from the Gaussian field, which is achieved by defining the Fourier transform $\hat F(x) \eqdef m + \hat N(x) \sqrt{\mu(x)} \ones_3$, where $N(p)$ is the realization of a Gaussian white noise, and $\sqrt{\cdot}$ is the matrix square root (see~\cite{galerne2011random} for more details).

Figure~\ref{fig:texsynth} shows an application where two input power spectra $(\mu,\nu)$ (computed using~\eqref{eq-power-spectrum} from two input textures exemplars $(f,g)$)  are interpolated using~\eqref{eq-interpolating}, and for each interpolation parameter $t \in [0,1]$ a new texture $F$ is synthesized and displayed.
Note that while the Q-Sinkhorn Algorithm~\ref{alg:sinkhorn} is provided for real PSD matrices, it extends verbatim to complex positive hermitian matrices (the matrix logarithm and exponential being defined the same way as for real matrices).


\newcommand{\TexSynthImg}[1]{\includegraphics[width=.195\linewidth]{texsynth/#1}}
\begin{figure}
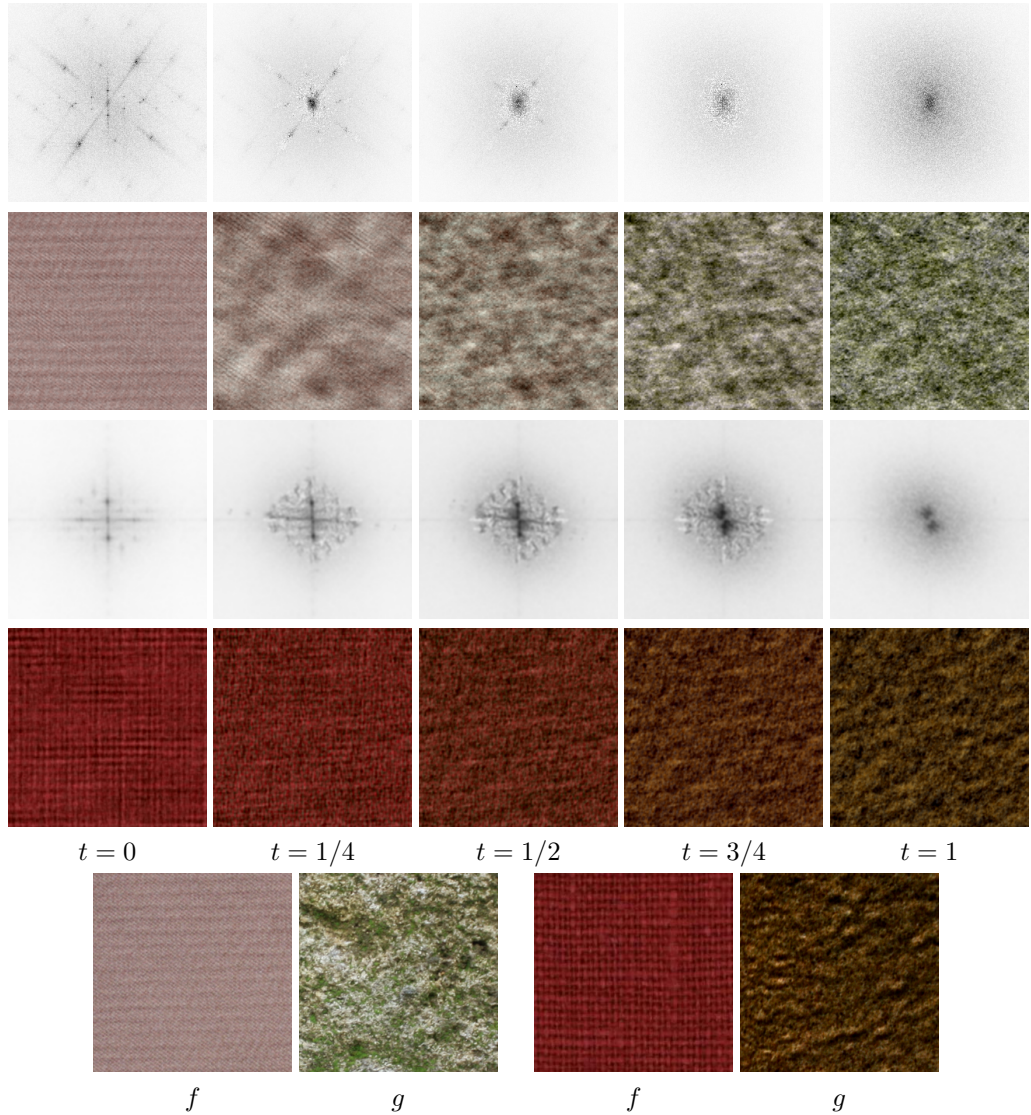
\centering
\begin{tabular}{@{}c@{\hspace{1mm}}c@{\hspace{1mm}}c@{\hspace{1mm}}c@{\hspace{1mm}}c@{}}
\TexSynthImg{synth1/spectrum-1}&
\TexSynthImg{synth1/spectrum-3}&
\TexSynthImg{synth1/spectrum-5}&
\TexSynthImg{synth1/spectrum-7}&
\TexSynthImg{synth1/spectrum-9}\\
\TexSynthImg{synth1/synthesis-1}&
\TexSynthImg{synth1/synthesis-3}&
\TexSynthImg{synth1/synthesis-5}&
\TexSynthImg{synth1/synthesis-7}&
\TexSynthImg{synth1/synthesis-9}\\
\TexSynthImg{synth2/spectrum-1}&
\TexSynthImg{synth2/spectrum-3}&
\TexSynthImg{synth2/spectrum-5}&
\TexSynthImg{synth2/spectrum-7}&
\TexSynthImg{synth2/spectrum-9}\\
\TexSynthImg{synth2/synthesis-1}&
\TexSynthImg{synth2/synthesis-3}&
\TexSynthImg{synth2/synthesis-5}&
\TexSynthImg{synth2/synthesis-7}&
\TexSynthImg{synth2/synthesis-9}\\
$t=0$ & $t=1/4$ & $t=1/2$ & $t=3/4$ & $t=1$
\end{tabular}
\begin{tabular}{@{}c@{\hspace{1mm}}c@{\hspace{5mm}}c@{\hspace{1mm}}c@{}}
\TexSynthImg{synth1/original-1}&
\TexSynthImg{synth1/original-2}&
\TexSynthImg{synth2/original-1}&
\TexSynthImg{synth2/original-2}\\
$f$  & $g$ & $f$  & $g$ 
\end{tabular}
\caption{\textbf{Row 1 and 3:}  display $\tr(\mu_t(x))$ where $\mu_t$ are the interpolated power spectra. 
\textbf{Rows 2 and 4:} realizations of the Gaussian field parameterized by the power spectra  $\mu_t$. 
\textbf{Row 5:} input texture exemplars from which $(\mu_{t=0},\mu_{t=1})=(\mu,\nu)$ are computed.
} \label{fig:texsynth}
\end{figure}


\section{Conclusion}

In this work, we have proposed a new static formulation for OT between tensor-valued measures. This formulation is an extension of the recently proposed unbalanced formulation of OT. A chief advantage of this formulation is that, once coupled with quantum entropic regularization, it leads to an effective numerical scheme, which is easily extended to the computation of barycenters. 

The proposed formulation is quite versatile, and can be extended to other convex cones beyond PSD matrices, which is a promising direction for future work and applications.
Other possible research directions also include investigating relationships between this static formulation of tensor-valued OT and dynamic formulations. 

\section*{Acknowledgements}

The work of Gabriel Peyr\'e has been supported by the European Research Council (ERC project SIGMA-Vision). 
J.\ Solomon acknowledges support of Army Research Office grant W911NF-12-R-0011 (``Smooth Modeling of Flows on Graphs'').
DTI data were provided by Franco Pestilli (NSF IIS 1636893; NIH ULTTR001108) and the Human Connectome Project (NIH 1U54MH091657). 
The images for the texture synthesis experiments are from IPOL.\footnote{\url{www.ipol.im}}
The images for the triangular meshing experiments are from Wikimedia.

\bibliographystyle{plain}
\bibliography{refs}

\end{document}